\documentclass[oribibl]{llncs}
\usepackage[ansinew]{inputenc}
\usepackage[T1]{fontenc}

\usepackage{tikz}
\usetikzlibrary{arrows,automata}

\usepackage{amsmath} 
\usepackage{amssymb}

\newcommand{\abs}[1]{\mbox{$ | #1 | $}}

\newcommand{\leftr}{\mathtt{left}}
\newcommand{\rightr}{\mathtt{right}}

\newcommand{\NN}{\mathcal{N}}
\newcommand{\PP}{{\mathcal{P}}}
\newcommand{\ZZ}{{\mathcal{Z}}}
\newcommand{\PPP}{{\mathbb{P}}}

\newcommand{\HH}{{\mathcal{H}}}
\newcommand{\cc}{{\mathfrak{c}}}
\newcommand{\SSS}{\mathcal{S}}

\newcommand{\BB}{\mathfrak{B}}
\newcommand{\BBB}{\mathfrak{B}^{\times}}

\newcommand{\mash}{\mathfrak{C}}
\def\nat{{\mathbb N}}
 \def\real{{\mathbb R}}

\def\vone{\mathbf{1}}
\newcommand{\norminf}[1]{\|#1\|_\infty}

\begin{document}

\title{Stochastic Context-Free Grammars,  Regular Languages, and Newton's Method} 
\author{Kousha Etessami\inst{1} \and Alistair Stewart\inst{1}   \and 
Mihalis Yannakakis\inst{2}}
\institute{School of Informatics, University of Edinburgh\\{\tt kousha@inf.ed.ac.uk} , {\tt stewart.al@gmail.com}
\and 
Department of Computer Science, Columbia University\\
{\tt mihalis@cs.columbia.edu}}

\maketitle

\begin{abstract}
We study the problem of computing the probability
that a given stochastic context-free grammar (SCFG), $G$, 
generates
a string in a given regular language $L(D)$ (given by a DFA, $D$).
This basic problem has a number of applications  in 
statistical natural language processing, and it is also a key necessary step towards
quantitative $\omega$-regular model checking of stochastic context-free
processes (equivalently, 1-exit recursive Markov chains, or stateless 
probabilistic pushdown processes).

We show that the probability that $G$ generates a string in $L(D)$ 
can be computed to within
arbitrary desired precision in polynomial time (in the standard
Turing model of computation), 
under a rather mild assumption about the SCFG, $G$,
and with no extra assumption about $D$.    
We show that this assumption is satisfied for SCFG's whose
rule probabilities are learned via the well-known
inside-outside (EM) algorithm for maximum-likelihood estimation
(a standard method for constructing SCFGs in
statistical NLP and biological sequence analysis).
Thus, for these SCFGs the algorithm always runs in P-time.

\end{abstract}

\thispagestyle{empty}

\vspace*{-0.3in}

\section{Introduction}

\vspace*{-0.1in}

{\em Stochastic} (or {\em Probabilistic) Context-Free Grammars} (SCFG) are
context-free grammars where the rules (productions) have associated probabilities.
They are a central stochastic model,
widely used in natural language processing \cite{ManSch99},
with applications also in biology (e.g. \cite{DEKM99,KH03}).
A SCFG $G$ generates a language $L(G)$ (like an ordinary CFG) and assigns a probability
to every string in the language. SCFGs have been extensively studied 
since the 1970's. A number of important problems on SCFGs can be viewed as
instances of the following {\em regular pattern matching problem} for different regular languages:

{\em Given a SCFG $G$ and a regular language $L$, given e.g., by a deterministic finite automaton (DFA) $D$,
compute the probability $\PPP_G(L)$ that $G$ generates a string in $L$, 
i.e. compute the sum
of the probabilities of all the strings in $L$.}

A simple example is when $L =\Sigma^*$, the set of all strings over the terminal alphabet $\Sigma$ of the SCFG $G$. 
Then this problem simply asks to compute the probability $\PPP_G(L(G))$ of the
language $L(G)$ generated by the grammar $G$.
Alternatively, if we view the SCFG as a stochastic process that starts from 
the start nonterminal, repeatedly applies the probabilistic rules
to replace (say, leftmost) nonterminals, and terminates when a string of terminals is reached,
then $\PPP_G(L(G))$ is simply the probability that this process terminates.
Another simple example is when $L$ is a singleton, $L=\{ w \}$, for some string $w$;
in this case the problem corresponds to the basic parsing question of 
computing the probability that a given string $w$ is generated by the SCFG $G$.
Another basic well-studied problem is the computation of {\em prefix probabilities}:
given a SCFG $G$ and a string $w$, compute the probability that $G$ generates a string
with prefix $w$ \cite{JelLaf91,St}. This is useful in online processing
in speech recognition \cite{JelLaf91} and corresponds to the 
case $L=w \Sigma^*$.
A more complex problem is the computation of {\em infix probabilities} \cite{Cor+,NS11-infix},
where we wish to compute the probability that $G$ generates a string that contains
a given string $w$ as a substring, which corresponds to the
language $L = \Sigma^* w \Sigma^*$.
In general, even when rule probabilities of the SCFG $G$ are rational,
the probabilities
we wish to compute can be irrational. Thus the typical aim 
for ``computing'' them is to approximate them to desired precision.

Stochastic context-free grammars are closely related to 
{\em 1-exit recursive Markov chains} (1-RMC) \cite{rmc}, and to 
{\em stateless probabilistic pushdown automata} (also called pBPA) \cite{EKM};
these are two equivalent models for a subclass of probabilistic programs with recursive
procedures.
The above regular pattern matching problem for SCFGs is equivalent to
the problem of computing the probability that a computation
of a given 1-RMC (or pBPA) terminates and satisfies a given regular property.
In other words, it corresponds to the quantitative model checking
problem for 1-RMCs with respect to regular {\em finite string} properties.

We first review some prior related work,
and then describe our results.

\medskip
\noindent
{\bf Previous Work.} 
As mentioned above, there has been, on the one hand, substantial work 
in the NLP literature
on different cases of the problem for various regular languages $L$, 
and on the other hand, there has been work 
in the verification and algorithms literature
on the analysis and model checking of recursive Markov chains
and probabilistic pushdown automata.
Nevertheless, even
the simple special case of $L=\Sigma^*$, the question
of whether it is possible to compute (approximately) in
polynomial time the desired probability  for a given SCFG $G$
(i.e. the probability $\PPP_G(L(G))$ of $L(G)$) was open until very recently.
In \cite{ESY12} we showed that $\PPP_G(L(G))$ can be computed
to arbitrary precision in polynomial time in the size of the input SCFG $G$
and the number of bits of precision. 
From a SCFG $G$, one can construct a multivariate system of equations $x=P_G(x)$,
where $x$ is a vector of variables and $P_G$ is a vector of polynomials
with positive coefficients which sum to (at most) 1. Such a system is called a 
{\em probabilistic polynomial system} (PPS), and it always has
a non-negative solution that is smallest in every coordinate, 
called the {\em least fixed point} (LFP).
A particular coordinate of the LFP of the system
$x=P_G(x)$ is the desired  probability $\PPP_G(L(G))$.
To compute $\PPP_G(L(G))$,  we used a variant of Newton's method on $x=P_G(x)$,
with suitable rounding after each step to control the bit-size of numbers,
and showed that it converges in P-time to the LFP \cite{ESY12}.
Building on this, we also showed that the probability 
$\PPP_G(\{w\})$ of string $w$
under SCFG $G$ can also be computed to any precision in P-time in the
size of $G$, $w$ and the number of bits of precision.

The use of Newton's method was proposed originally in \cite{rmc} 
for computing termination probabilities for (multi-exit) RMC's,
which requires 
the solution of equations from a more general class of polynomial systems $x=P(x)$,
called {\em monotone polynomial systems} (MPS), where the
polynomials of $P$ have positive coefficients, but their sum is not restricted
to $\leq 1$.
An arbitrary MPS may not have any non-negative solution, 
but if it does then it has a LFP, 
and a version of Newton provably converges to the LFP \cite{rmc}. 
There are now implementations of variants of Newton's method in
several tools \cite{WojEte07,NedSat08} and experiments 
show that they perform well on many instances. 
The rate of convergence of Newton for general MPSs 
was studied in detail in \cite{lfppoly}, and was further 
studied most recently in \cite{ESY13} (see below). 
In certain cases, Newton converges fast,
but in general there are exponential bad examples.
Furthermore, there are negative results indicating it is very unlikely
that any non-trivial approximation of termination probabilities of multi-exit RMCs,
and the LFP of MPSs, 
can 
be done in P-time (see \cite{rmc}).

The model checking problem for RMCs (equivalently pPDAs) and
$\omega$-regular properties was studied in \cite{EKM,EY-MC-12}.
This is of course a more general problem than the problem for SCFGs
(which correspond to 1-RMCs) and regular languages (the finite string case
of $\omega$-regular languages).
It was shown in \cite{EY-MC-12} that in the case of 1-RMCs,
the qualitative problem of determining whether the probability
that a run satisfies the property is 0 or 1 
can be solved in P-time in the size of the 1-RMC, but for the quantitative
problem of approximating the probability, the algorithm runs in PSPACE,
and no better complexity bound was known. 

The particular cases of computing prefix and infix probabilities for a SCFG
have been studied in the NLP literature, but
no polynomial time algorithm for general SCFGs is known. 
Jelinek and Lafferty gave an algorithm for
grammars in Chomsky Normal Form (CNF) \cite{JelLaf91}. 
Note that a general SCFG $G$ may not have any equivalent CNF grammar with rational rule probabilities,
thus one can only hope for an ``approximately equivalent" CNF grammar;
constructing such a grammar in the
case of stochastic grammars $G$ is non-trivial, at least as difficult as
computing the probability of $L(G)$, and the first P-time algorithm
was given in \cite{ESY12}. Another algorithm for prefix probabilities
by Stolcke \cite{St} applies to general SCFGs, but in the presence of unary
and $\epsilon$-rules, the algorithm does not run in polynomial time.
The problem of computing infix probabilities 
was studied in \cite{Cor+,NedSat08,NS11-infix},
and in particular \cite{NedSat08,NS11-infix} cast it in the
general regular language framework, and studied the 
general problem of computing the probability $\PPP_G(L(D))$
of the language $L(D)$ of a DFA $D$ under a SCFG $G$.
From $G$ and $D$ they construct a product {\em weighted context-free grammar}
(WCFG) $G'$: a CFG with (positive) weights on the rules, which may not be probabilities,
in particular the weights on the rules of a nonterminal may sum to more than 1.
The desired probability $\PPP_G(L(D))$ is the weight of $L(G')$. 
As in the case of SCFGs, this weight is given by the LFP of a monotone system of equations $y=P_{G'}(y)$, however, unlike the case of SCFGs the system now
is not a probabilistic system (thus our result of \cite{ESY12} does not apply).
Nederhof and Satta then solve the system using the decomposed Newton method from 
\cite{rmc} and Broyden's (quasi-Newton) method, and present experimental results
for infix probability computations.

Most recently, in \cite{ESY13}, we have obtained worst-case upper bounds
on (rounded and exact) Newton's method applied to arbitrary MPSs, $x=P(x)$,
as a function of the input encoding size $|P|$ and $\log (1 /\epsilon)$, to converge to
within additive error $\epsilon > 0$ of the LFP solution $q^*$.   
However, our bounds in \cite{ESY13}, even when $0 < q^* \leq 1$, 
are exponential in the depth 
of (not necessarily critical) strongly connected components of $x=P(x)$, and furthermore
they also depend linearly on $\log(\frac{1}{q^*_{\min}})$, where $q^*_{\min} = \min_i q^*_i$, which can be $\approx \frac{1}{2^{2^{|P|}}}$.
As we describe next, we do far better in this paper for the  MPSs that arise from the ``product'' 
of a SCFG and a DFA.

\medskip
\noindent
{\bf Our Results.}
We study the general problem of computing the probability $\PPP_G(L(D))$ that a given SCFG $G$ generates a string in the language $L(D)$ of a given DFA $D$.
We show that, under a certain mild assumption on $G$, this probability can be computed 
to any desired precision in time polynomial in the encoding sizes of $G$ \& $D$ and the number of bits of precision. 

We now sketch briefly the approach and state the assumption on $G$.
First we construct from $G$ and $D$ the product weighted CFG $G' = G \otimes D$
as in \cite{NedSat08} and construct the corresponding MPS $y=P_{G'}(y)$,
whose LFP contains the desired probability $\PPP_G(L(D))$ as one of its components.%
The system is monotone but not probabilistic.
We eliminate (in P-time) those variables that have value 0 in the LFP, 
and apply Newton, with
suitable rounding in every step.
The heart of the analysis shows there is a tight algebraic correspondence
between the behavior of Newton's method on this MPS and its behavior on the
probabilistic polynomial system (PPS) $x=P_G(x)$ of $G$.
In particular, this correspondence shows that, with exact arithmetic,
the two computations converge at the same rate.
By exploiting this, and by extending recent results we established for PPSs, we obtain the conditional polynomial upper bound.
Specifically, call a PPS $x=P(x)$ {\em critical} if the spectral radius of the Jacobian of $P(x)$,
evaluated at the LFP $q^*$ is equal to 1 (it is always $\leq 1$).
We can
form a dependency graph between the variables of a PPS, and decompose the variables and
the system into strongly connected components (SCCs); an SCC is called critical
if the induced subsystem on that SCC is critical.
The {\em critical depth} of a PPS is the maximum number of critical SCCs on any
path of the DAG of SCCs (i.e. the max nesting depth of critical SCCs).
We show that if the PPS of the given SCFG $G$ has bounded (or even logarithmic) 
critical depth, then we can compute $\PPP_G(L(D))$ (for any DFA $D$) 
in polynomial time in the size of $G$, $D$ and the number of bits of precision.%

Furthermore, we show this condition is satisfied by a broad class of SCFGs
used in applications. Specifically, a standard way the probabilities of 
rules of a SCFG are set is by using the EM (inside-outside) algorithm. We show that the
SCFGs constructed in this way are guaranteed to be noncritical (i.e., have critical
depth 0). So for these SCFGs, and any DFA, the algorithm runs in P-time.

The paper is organized as follows.
Section 2 gives definitions and background. Section 3
establishes tight algebraic connections between the behavior of Newton
on the PPS of the SCFG, and on the MPS of the product WCFG.  Section 4 
proves the claimed bounds on rounded Newton's method. Section 5 shows the
noncriticality of SCFGs obtained by the EM method. 
Proofs are in the Appendix.

\vspace*{-0.13in}

\section{Definitions and Background}

\vspace*{-0.1in}

A {\em weighted context-free grammar} (WCFG), 
$G =  (V,\Sigma,  R, p)$, has
a finite set $V$ of {\em nonterminals},  
a finite set $\Sigma$ of {\em terminals} (alphabet symbols),
and a finite list of {\em rules}, $R \subset V \times (V \cup \Sigma)^*$,
where each rule $r \in R$ is a pair $(A,\gamma)$,
which we usually denote by $A \rightarrow \gamma$,
where $A \in V$ and $\gamma \in (V \cup \Sigma)^*$.
Finally $p: R \rightarrow \real^+$ maps each rule $r \in R$ to a positive
{\em weight}, $p(r) > 0$.  
We often denote a rule $r = (A \rightarrow \gamma)$  together with its weight
by writing $A \stackrel{p(r)}{\rightarrow} \gamma$.
We will sometimes also specify a specific non-terminal $S \in V$ as the starting symbol.

Note that we allow $\gamma \in (V \cup \Sigma)^*$ to possibly be 
the  empty string, denoted by $\epsilon$.
A rule 
of the form $A {\rightarrow} \epsilon$ is called an {\em $\epsilon$-rule}.
For a rule $r = (A \rightarrow \gamma)$, we 
let $\leftr(r) := A$ and $\rightr(r) := \gamma$. 
We let $R_A = \{ r \in R \mid \leftr(r) = A\}$.
For $A \in V$, let $p(A) = \sum_{r \in R_A} p(r)$.
A WCFG, $G$, is called a {\em stochastic} or {\em probabilistic}
{\em context-free grammar}  (SCFG or PCFG; we shall use SCFG), if 
for $\forall  A \in V$, $p(A)  \leq 1$.
An SCFG is called {\em proper}
if $\forall A \in V, \; p(A) = 1$.

We will say that an WCFG, $G=(V,\Sigma,R,p)$ is in 
{\em Simple Normal Form} (SNF) if every nonterminal $A \in V$ belongs to one of the following three types:
\vspace*{-0.06in}
\begin{enumerate}
\item type {\tt L}: every rule $r \in R_A$,  has the form $A \xrightarrow{p(r)} B$.
\item type {\tt Q}: there is a single rule in $R_A$: $A \xrightarrow{1} BC$, 
for some $B, C \in V$.
\item type {\tt T}: there is a single rule in $R_A$:  either 
$A  \xrightarrow{1} \epsilon$, 
or $A  \xrightarrow{1} a$ for some $a \in \Sigma$.
\end{enumerate}

For a WCFG, $G$, 
strings $\alpha, \beta \in (V \cup \Sigma)^*$, and 
$\pi = r_1 \ldots r_k \in R^*$, we write 
$ \alpha \stackrel{\pi}{\Rightarrow} \beta$
if the leftmost derivation starting from $\alpha$, and  
applying the sequence $\pi$ of rules, derives $\beta$.
We let $p( \alpha \stackrel{\pi}{\Rightarrow} \beta) = 
\prod^k_{i=1} p(r_k)$ if $\alpha \stackrel{\pi}{\Rightarrow} \beta$,
and $p(\alpha \stackrel{\pi}{\Rightarrow} \beta) = 0$ otherwise.
If $A \stackrel{\pi}{\Rightarrow} w$ for $A \in V$ and $w \in \Sigma^*$,
we say that $\pi$ is a {\em complete} derivation from $A$
and its {\em yield} is $y(\pi) =w$.
There is a natural one-to-one correspondence between the complete derivations
of $w$ starting at $A$ and the {\em parse trees} of $w$
rooted at $A$, and this correspondence preserves 
weights.

For a WCFG, $G = (V,\Sigma,R,p)$,
nonterminal $A \in V$, 
and terminal string $w \in \Sigma^*$, we let
$p_A^{G,w} = \sum_{\{\pi \mid y(\pi) = w\}} p(A \stackrel{\pi}{\Rightarrow} w)$.
For a general WCFG, $p_A^{G,w}$ need not be a finite value (it may be $+\infty$, since
the sum may not converge).
Note however that if $G$ is an SCFG, then
$p_A^{G,w}$ defines the probability that, starting
at nonterminal $A$, $G$ generates $w$, and thus it is clearly finite.

The {\em termination probability} ({\em termination weight}) of an SCFG (WCFG), $G$,
starting at nonterminal $A$, denoted $q_{A}^{G}$,
is defined by  $q_A^{G} = \sum_{w \in \Sigma^*}  p_A^{G,w}$.
Again, for an arbitrary WCFG $q_A^{G}$ need not be a finite number.
A WCFG $G$ is called {\em convergent} if $q_A^{G}$ is finite for all $A \in V$.
We will only encounter convergent WCFGs in this paper,
so when we say WCFG we mean convergent WCFG, unless otherwise specified.
In $G$ is an SCFG, then $q_A^G$ is just  
the total probability with which the derivation
process starting at $A$ eventually generates a finite string and (thus) stops,
so SCFGs are clearly convergent.

An SCFG, $G$, is called 
{\em consistent starting at $A$}
if $q_A^G = 1$, and  $G$ is called {\em consistent} if it is consistent
starting at every nonterminal.
Note that even if a SCFG, $G$, is proper
this does not necessarily imply that $G$ is 
consistent.  
For an SCFG, $G$, we can decide whether $q_A^G = 1$ in P-time (\cite{rmc}).
The same decision problem is PosSLP-hard for convergent WCFGs (\cite{rmc}).

For any WCFG, $G = (V,\Sigma,R,p)$,
with $n = |V|$,
assume the nonterminals in $V$ are indexed as $A_1, \ldots, A_n$.
We define the following 
{\bf\em monotone polynomial system of equations} (MPS)
associated with $G$, denoted  $x = P_G(x)$.
Here $x = (x_1, \ldots,x_n)$ denotes an $n$-vector of variables. 
Likewise $P_G(x) = (P_G(x)_1, \ldots, P_G(x)_n)$ denotes an
$n$-vector of multivariate polynomials over the variables $x = (x_1,\ldots,x_n)$.
For a vector $\kappa = (\kappa_1, \kappa_2, \ldots, \kappa_n) \in \nat^n$,
we use the notation $x^{\kappa}$ to denote the monomial $x_1^{\kappa_1} x_2^{\kappa_2} \ldots x_n^{\kappa_n}$.
For a  non-terminal $A_i \in V$, and a string $\alpha \in (V \cup \Sigma)^*$,
let $\kappa_{i}(\alpha) \in \nat$ denote the number of occurrences of $A_i$ in
the string $\alpha$.  
We define $\kappa(\alpha) \in \nat^n$ to be  
$\kappa(\alpha) = (\kappa_{1}(\alpha), 
\kappa_{2}(\alpha), \ldots, \kappa_{n}(\alpha))$.

In the MPS  $x = P_G(x)$,  corresponding to each nonterminal $A_i \in V$,
there will be one variable $x_{i}$ 
and one equation, namely  $x_{i} = P_{G}(x)_{i}$,
where: %
$P_G(x)_i \equiv  \sum_{r = (A \rightarrow \alpha) \in R_{A_i}}  p(r)  x^{\kappa(\alpha)}$.
If there are no rules associated with $A_i$, i.e., if $R_{A_i} = \emptyset$, then by default we define $P_G(x)_{i} \equiv 0$.
Note that if $r \in R_{A_i}$ is a terminal rule, i.e., $\kappa(r) = (0, \ldots,0)$, 
then $p(r)$  is one of the constant terms of $P_G(x)_i$.

{\bf Note:}  {\em Throughout this paper, for any $n$-vector $z$, whose $i$'th 
coordinate $z_i$ ``corresponds'' to nonterminal $A_i$, 
we often find it convenient to use $z_{A_i}$ to refer to $z_i$.}
So, e.g., we alternatively use $x_{A_i}$ and $P_{G}(x)_{A_i}$,  instead of 
$x_i$ and $P_G(x)_i$. 

Note that if $G$ is a SCFG, then in $x = P_G(x)$,
by definition, the sum of the monomial coefficients and constant terms of 
each polynomial $P_G(x)_i$ is at most $1$, because $\sum_{r \in R_{A_i}} p(r) \leq 1$
for every 
$A_i \in  V$. 
An MPS that satisfies this extra condition
is called a {\bf\em probabilistic polynomial system of equations} (PPS).

Consider any MPS, 
$x = P(x)$, with $n$ variables, $x=(x_1,\ldots,x_n)$.
Let $\real_{\geq 0}$ denote the non-negative real numbers.
Then $P(x)$ defines a monotone operator
on the non-negative orthant $\real^n_{\geq 0}$. 
In general, an MPS  need not have any real-valued solution: consider $x= x+1$.
However, by monotonicity of $P(x)$, if there exists $a \in \real^n_{\geq 0}$
such that $a=P(a)$, then there is a {\em least fixed point} (LFP) solution 
$q^* \in \real^n_{\geq 0}$
such that $q^* = P(q^*)$, and such that $q^* \leq a$  for all
solutions $a \in \real^{n}_{\geq 0}$.

\begin{proposition}(cf. \cite{rmc} or see \cite{NedSat08b})
For any SCFG (or convergent WCFG),  $G$, with $n$ nonterminals
$A_1, \ldots, A_n$,  
the LFP solution of $x = P_G(x)$
is the $n$-vector $q^G = (q^G_{A_1}, \ldots, q^G_{A_n})$ of
termination probabilities (termination weights) of $G$. 
\label{prop:lfp-wcfg-is-term}
\end{proposition}

For computation purposes,
we assume that the input probabilities (weights) associated with rules of
input SCFGs or WCFGs are positive rationals encoded 
by giving their numerator and denominator in binary.
We use $|G|$ to denote the encoding size (i.e., number of bits) of a input WCFG  $G$.

Given any WCFG (SCFG) $G=(V,\Sigma, R,p)$ we can compute in linear time
an SNF form WCFG (resp. SCFG) $G'=(V'\Sigma,R',p')$ of size $|G'|=O(|G|)$
with $V' \supseteq V$ such that $q_A^{G,w} = q_A^{G',w}$ for all
$A \in V$, $w \in \Sigma^*$ (cf. \cite{rmc} and Proposition 2.1 of \cite{ESY12}).
Thus, for the problems studied in this paper, we may assume wlog
that a given input WCFG or SCFG is in SNF form.

A DFA, $D=(Q,\Sigma,\Delta,s_0,F)$, has
states $Q$, alphabet $\Sigma$, transition function $\Delta: Q \times \Sigma \rightarrow Q$, 
start state $s_0 \in Q$ and final states $F \subseteq Q$.
We extend $\Delta$ to strings:  $\Delta^*: Q \times \Sigma^* \rightarrow Q$ 
is defined by induction on the length $|w| \geq 0$ of $w \in \Sigma^*$:
for $s \in Q$,  $\Delta^*(s, \epsilon) := s$.
Inductively, if $w = a w'$, with $a \in \Sigma$, then $\Delta^*(s,w) := 
\Delta^*(\Delta(s,a),w')$.
We define $L(D) = \{ w \in \Sigma^* \mid  \Delta^*(s_0,w) \in F \}$.

Given a WCFG $G$ and a DFA $D$ 
over the same terminal alphabet,  
for any nonterminal $A$ 
of $G$, we define 
$q_A^{G,D} = \sum_{w \in L(D)} q_A^{G,w}$.
If $G$ is a SCFG, $q_A^{G,D}$ 
simply denotes the
probability that $G$, starting at $A$, 
generates a string in $L(D)$.
Our goal is to compute $q_A^{G,D}$, given SCFG $G$ and DFA $D$.
In general, $q_A^{G,D}$ may be an irrational probability,
even when all of the rule probabilities of $G$ are rational
values.   So one natural goal is to approximate $q_A^{G,D}$
to within desired precision.
More precisely, the approximation problem
is this: given as input an SCFG, $G$, with a specified nonterminal $A$,
a DFA, $D$,   
over the  same terminal alphabet $\Sigma$,
and a rational error threshold  $\delta > 0$,
output a rational value $v \in [0,1]$ 
such that $| v - q_A^{G,D} | < \delta$.
We would like to do this as efficiently as possible
as a function of the input size: $|G|$, $|D|$, and $\log(1/\delta)$.

To compute $q_A^{G,D}$, it will be useful to 
define a WCFG obtained as the {\em product} of
a SCFG and a DFA.  
We assume, wlog, that the input SCFG 
is in SNF form. 
The {\bf\em product} (or {\bf\em intersection}) of a SCFG 
$G = (V,\Sigma,R,p)$ in SNF form,
and DFA, $D=(Q,\Sigma,\Delta,s_0,F)$,
is defined to be 
a new WCFG,  $G \otimes D = ( V',  
\Sigma, R', p')$,
where  the set of nonterminals is 
$V' = Q \times V \times Q$. 
Assuming  $n = |V|$ and $d = |Q|$,  then   $|V'| = d^2n$.
The rules $R'$ and rule probabilities $p'$ of the product $G \otimes D$ are defined  as follows
(recall $G$ is assumed to be in SNF):
\vspace*{-0.10in}
\begin{itemize}
\item Rules of form {\tt L}:  For every rule of the form $(A \xrightarrow{p} B) \in R$,
and every pair of states $s,t \in Q$,  there is a rule
$(sAt) \xrightarrow{p} (sBt)$ in $R'$. 

\item Rules of form {\tt Q}:  for every rule
$(A \xrightarrow{1} BC) \in R$, and for all states 
$s, t, u \in Q$, there is a rule $(sAu) \xrightarrow{1} (sBt)(tCu)$ in $R'$. 

\item Rules of form {\tt T}: for every rule  $(A \xrightarrow{1} a) \in R$,
where $a \in \Sigma$,
and for every state $s \in Q$, if $\Delta(s,a) = t$,
then there is a rule $(sAt) \xrightarrow{1} a$ in $R'$.

For every rule $(A \xrightarrow{1} \epsilon) \in R$, 
and every $s \in Q$, there is a rule $(sAs) \xrightarrow{1} \epsilon$
\end{itemize}
Associated with the WCFG, $G \otimes D$, is the MPS 
$y = P_{G \otimes D}(y)$,  where $y$ is now a $d^2n$-vector of variables,
where $n = |V|$ and $d = |Q|$.
The LFP solution of this MPS
captures the probabilities $q_A^{G,D}$ 
in the following sense:

\begin{proposition}{(cf. \cite{NS11-infix}, or \cite{EY-MC-12} for a variant of this)}
For any SCFG, $G = (V,\Sigma,R,p)$, and DFA, $D = (Q,\Sigma,\Delta,s_0,F)$, 
the LFP solution $q^{G \otimes D}$ of the MPS  
$x = P_{G \otimes D}(x)$,  satisfies
${\mathbf 0} \leq q^{G \otimes D} \leq {\mathbf 1}$.
Furthermore, for any $A \in V$ and $s,t \in Q$, 
$ q^{G \otimes D}_{(sAt)} = \sum_{\{ w \mid \Delta^*(s,w) = t\}}   q_A^{G,w}$.
Thus,  for every $A \in V$,  $q_A^{G,D} = \sum_{t \in F} q^{G \otimes D}_{(s_0At)}$.
\label{prop:lfp-prod-scfg-is-term}
\end{proposition}

\noindent{\bf Newton's method (NM).}
For an MPS (or PPS), $x=P(x)$,  in $n$ variables,
let $B(x) := P'(x)$ denote the Jacobian matrix of $P(x)$.
In other words, $B(x)$ is an $n \times n$ matrix such that 
$B(x)_{i,j} = \frac{\partial P_i(x)}{\partial x_j}$.
\noindent For a vector $z \in \real^n$, 
assuming that matrix $(I - B(z))$ is non-singular, 
we define a single iteration of Newton's method (NM) for $x=P(x)$ on $z$ 
via the following operator:
\begin{equation}\label{newton-one-it-eq}
\mathcal{N}(z) :=     z + (I-B(z))^{-1} (P(z) - z)
\end{equation}
Using Newton iteration, starting at $n$-vector $x^{(0)} :=  {\bf 0}$,
yields the following iteration:
$x^{(k+1)} :=  \mathcal{N}(x^{(k)})$,  for $k = 0, 1, 2, \ldots$.

For every MPS, we can detect in P-time 
all the variables $x_j$ such that $q^*_j = 0$ \cite{rmc}.
We can then remove these variables
and  their corresponding equation $x_j = P(x)_j$,
and substitute their values on the right hand sides 
of remaining equations.
This yields a new MPS, with LFP $q' > 0$,
which corresponds to 
the non-zero coordinates of $q^*$.
It was shown in \cite{rmc} that one can always apply 
a decomposed Newton's method to this MPS, to converge
monotonically to the LFP solution.  

\begin{proposition}{(cf.  Theorem 6.1 of \cite{rmc} and Theorem 4.1 of \cite{lfppoly})}
\label{prop:monotone-conv-newt}
Let $x=P(x)$ be a MPS, with LFP 
$q^* > {\textbf 0}$.
Then starting at $x^{(0)} := \textbf{0}$, the Newton iterations $x^{(k+1)} := 
\mathcal{N}(x^{(k)})$ are  well
defined and monotonically converge to $q^*$, i.e.
$\lim_{k \rightarrow \infty}  x^{(k)} = q^*$,
and $x^{(k+1)} \geq x^{(k)} \geq \textbf{0}$ for all $k \geq 0$.
\end{proposition}

Unfortunately, 
it was shown in \cite{rmc} that 
obtaining any non-trivial additive approximation to the LFP solution of a general MPS,
even one whose LFP is $0 < q^* \leq 1$,  is {\bf PosSLP}-hard,
so we can not compute the termination weights of general WCFGs in P-time (nor even in NP), 
without a major breakthrough in the complexity of numerical computation.  (See \cite{rmc} for more information.)

Fortunately, for the class of PPSs, we can do a lot better. 
First we can identify in P-time also all the variables $x_j$ such that 
$q^*_j = 1$ \cite{rmc} and remove them from the system.
We showed recently in \cite{ESY12}
that by then applying a suitably {\em rounded down} variant of Newton's method 
to the resulting PPS,
we can approximate $q^*$ within additive error $2^{-j}$
in time polynomial in the size of the PPS and $j$.

\vspace*{-0.12in}

\section{Balance, Collapse, and Newton's method}

\vspace*{-0.05in}

For an SCFG, $G = (V,\Sigma,R,p)$, and a DFA, $D = (Q,\Sigma,\Delta,s_0,F)$,
we want to relate the behavior of Newton's method on the MPS associated with
the WCFG,  $G \otimes D$, to that of the 
PPS associated with the SCFG $G$. 
We shall show that there is indeed a tight correspondence, regardless of what the DFA $D$ is.
This holds even when $G$ itself is a convergent WCFG, and thus $x=P_G(x)$ is
an MPS.
We need an abstract algebraic way to express this correspondence.
A key notion will be {\em balance}, and the {\em collapse} operator defined on balanced vectors and matrices.

Consider the LFP $q^G$  of $x = P_G(x)$, and
LFP $q^{G \otimes D}$ of  $y = P_{G \otimes D}(y)$.
By 
Propos. \ref{prop:lfp-wcfg-is-term}
and \ref{prop:lfp-prod-scfg-is-term}, for any $A \in V$,
$q^G_A = \sum_{w \in \Sigma^*} q_A^{G,w}$  is the probability (weight) 
that $G$, starting at $A$, generates any finite string.
Likewise $q^{G \otimes D}_{(sAt)} = \sum_{\{w \mid  \Delta^*(s,w) = t \}} q_A^{G,w}$  is the probability (weight) that, starting at $A$,  
$G$ generates a finite string $w$ such that $\Delta^*(s,w) = t$. 
Thus, for any $A \in V$ and $s \in Q$,  $q^G_A = \sum_{t \in Q} q^{G \otimes D}_{(sAt)}$.

It turns out that analogous relationships hold between many other vectors
associated with $G$ and $G \otimes D$, 
including between the Newton iterates 
obtained by applying Newton's method to their respective PPS (or MPS) and 
the product MPS.
Furthermore, associated relationships also hold between
the Jacobian matrices $B_G(x)$ and $B_{G \otimes D}(y)$
of $P_G(x)$ and $P_{G \otimes D}(y)$, respectively.

Let $n = |V|$ and let $d = |Q|$.
A vector $y \in \real^{d^2n}$,
whose coordinates are indexed by triples $(sAt) \in Q \times V \times Q$, 
is called {\bf\em balanced} if for any
non-terminal $A$, and any pair of states $s, s' \in Q$,
 $\sum_{t \in Q} y_{(sAt)} =  \sum_{t \in Q} y_{(s'At)}$.
In other words, $y$ is balanced if the value of the sum $\sum_{t \in Q} y_{(sAt)}$  is independent of the state $s$.  
As already observed, $q^{G \otimes D} \in \real^{d^2n}_{\geq 0}$ is balanced.
Let  $\BB \subseteq \real^{d^2n}$ denote the set of balanced vectors.
Let us define the {\bf\em collapse} mapping $\mash: \BB  \rightarrow \mathbb{R}^n$.
For any $A \in V$, 
$\mash(y)_A := \sum_t y_{(sAt)}$.
Note: $\mash(y)$ is well-defined, because for $y \in \BB$,
and any $A \in V$,   the sum $\sum_t y_{(sAt)}$ 
is by definition independent of the state $s$.  

We next extend the definition of balance to matrices.
A matrix  $M \in \mathbb{R}^{d^2n \times d^2n}$  is called {\bf\em balanced}  if, for any non-terminals 
$B, C \in V$ and states $s, u \in Q$, 
and for any  pair of states $v, v' \in Q$,
$\sum_t M_{(sBt), (uCv)} = \sum_t M_{(sBt), (uCv')}$, 
and for any $s,v \in Q$ and $s',v' \in Q$,
 $\sum_{t,u} M_{(sBt),(uCv)}  = 
\sum_{t,u} M_{(s'Bt),(uCv')}$. 
Let $\BBB \subseteq \real^{d^2n \times d^2n}$ denote the set of balanced matrices.
We extend the {\bf\em collapse} map $\mash$ to matrices.  
$\mash: \BBB \rightarrow \real^{n \times n}$ is defined
as follows.  For any $M \in \BBB$, and any 
$B, C \in V$,
$\mash(M)_{BC} := \sum_{t,u} M_{(sBt), (uCv)}$.
Note, again, $\mash(M)$ is well-defined. 

We denote the Newton operator, $\NN$, applied to a vector $x' \in \real^n$ for the PPS $x=P_G(x)$ associated with $G$ by 
$\mathcal{N}_{G}(x')$.  Likewise, we denote the Newton
operator applied to a vector $y' \in \real^{d^2n}$ for the MPS $y=P_{G \otimes D}(y)$ associated with $G \otimes D$
by  $\mathcal{N}_{G \otimes D}(y')$.
For a real square matrix  $M$,  let $\rho(M)$ denote the spectral radius of $M$.
The main result of this section is the following:
 
\begin{theorem} \label{balthm} 
Let $x=P_G(x)$ be any PPS (or MPS), with $n$ variables, associated with a SCFG (or WCFG) $G$, 
and let $y = P_{G \otimes D}(y)$ be the corresponding product MPS,  
for any DFA $D$, with $d$ states.   For any balanced vector
$y \in \BB \subseteq \mathbb{R}^{d^2n}$, with $y \geq 0$,
$\rho(B_{G \otimes D}(y)) = \rho(B_{G}(\mash(y)))$.  
Furthermore, if $\rho(B_{G \otimes D}(y)) < 1$, 
then $\NN_{G \otimes D}(y)$ is defined and balanced, 
$\NN_{G}(\mash(y))$ is defined, and 
$\mash(\NN_{G \otimes D}(y)) = \NN_G(\mash(y))$.
Thus, $\NN_{G \otimes D}$ preserves balance, and
the collapse map $\mash$  ``commutes'' with $\NN$ over non-negative balanced vectors, irrespective of what the DFA $D$ is.
\end{theorem}
We prove this in the appendix via a series of lemmas that reveal
many algebraic/analytic properties of balance, collapse, and 
Newton's method.  
Key is:

\begin{lemma} \label{bal-lem1}
Let $\BB_{\geq 0} = \BB \cap \real^{d^2n}_{\geq 0}$ and $\BBB_{\geq 0} = \BB \cap \real^{d^2n \times d^2n}_{\geq 0}$.\\
We have $q^{G \otimes D} \in \BB_{\geq 0}$ and $\mash(q^{G \otimes D}) = q^G$, and:
\begin{itemize}
\item[(i)] If $y \in \BB_{\geq 0} \subseteq \real^{d^2n}_{\geq 0}$  
then $B_{G \otimes D}(y) \in \BBB_{\geq 0}$,  
and $\mash(B_{G \otimes D}(y)) = B_G(\mash(y))$.

\item[(ii)] If $y \in \BB_{\geq 0}$, then $P_{G \otimes D}(y) \in \BB_{\geq 0}$, 
and
$\mash(P_{G \otimes D}(y)) = P_G(\mash(y))$.

\item[(iii)] If  $y \in \BB_{\geq 0}$ and 
$\rho(B_G(\mash(y))) < 1$, then $I-B_{G \otimes D}(y)$ is non-singular,\\ 
$(I-B_{G \otimes D}(y))^{-1} \in \BBB_{\geq 0}$, 
and $\mash((I-B_{G \otimes D}(y))^{-1}) = (I - B_G(\mash(y)))^{-1}$.

\item[(iv)] If $y \in \BB_{\geq 0}$ and $\rho(B_G(\mash(y))) < 1$, then 
$\NN_{G \otimes D}(y) \in \BBB$\\ and $\mash(\NN_{G \otimes D}(y)) = \NN_G(\mash(y))$.
\end{itemize}
\end{lemma}

An easy consequence of Thm. \ref{balthm} (and Prop. \ref{prop:monotone-conv-newt}) 
is that if we use NM with exact arithmetic on the PPS or MPS, $x=P_G(x)$,
and on the product MPS, $y = P_{G \otimes D}(y)$, they converge at the same rate:

\begin{corollary} 
\label{cor:same-rate-newton}
For any PPS or MPS, $x=P_G(x)$, with LFP $q^G > 0$, 
and corresponding product MPS,  $y = P_{G \otimes D}(y)$,  if we  use Newton's method
with {\em exact arithmetic}, starting
at $x^{(0)} := 0$, and $y^{(0)} := 0$,  then
all the Newton iterates $x^{(k)}$ and $y^{(k)}$ are well-defined, and for all $k$:
$\quad$ $x^{(k)} = \mash(y^{(k)})$.
\end{corollary}

\vspace*{-0.2in}

\section{Rounded Newton on PPSs and product MPSs}

\vspace*{-0.1in}

To work in the Turing model of computation (as opposed to the unit-cost RAM model)
we have to consider {\em rounding} between iterations of NM, as in \cite{ESY12}.

\begin{definition} {({\bf Rounded-down Newton's method} ({\bf R-NM}), with parameter $h$.)}
Given an MPS, $x=P(x)$,  
with LFP $q^*$, 
where $q^* > {\textbf 0}$,
in R-NM with integer 
rounding parameter $h > 0$, 
we compute a sequence of 
iteration vectors $x^{[k]}$.  Starting  with $x^{[0]} 
:= \mathbf{0}$, $\forall k \geq 0$ we compute 
$x^{[k+1]}$ as follows:
\vspace*{-0.06in}
\begin{enumerate}
\item  Compute $x^{\{k+1\}} :=  \mathcal{N}_P(x^{[k]})$, 
where $\mathcal{N}_P(x)$ is the Newton
op. defined in (\ref{newton-one-it-eq}).

\item For each coordinate $i=1,\ldots,n$, set $x^{[k+1]}_i$ 
to be equal 
to the maximum multiple of $2^{-h}$ which is $\leq \max (x^{\{k+1\}}_i, 0)$.
(In other words, round down $x^{\{ k+1\}}$ to the nearest multiple of 
$2^{-h}$, while ensuring the result is non-negative.)
\end{enumerate}
\end{definition}

\noindent Unfortunately, rounding can cause iterates $x^{[k]}$ to become unbalanced. 
Nevertheless, we can handle this.
For any PPS, $x=P(x)$, 
with Jacobian matrix $B(x)$, and LFP $q^*$,  $\rho(B(q^*)) \leq 1$  (\cite{rmc,ESY12}).
If $\rho(B(q^*)) < 1$, we call the PPS {\bf non-critical}.  Otherwise, if $\rho(B(q^*)) = 1$, we 
call the PPS {\bf critical}.    
For SCFGs whose PPS $x= P_G(x)$ is non-critical, we 
get good bounds, even though R-NM iterates can become unbalanced:

\begin{theorem} 
\label{poly-iterations-non-critical} 
For any $\epsilon > 0$, and for an SCFG, $G$, if the PPS $x= P_G(x)$  has LFP $0 < q^G \leq 1$ and $\rho(B_G(q^G)) < 1$, 
then if we use R-NM with parameter $h+2$ to approximate the 
LFP solution of the MPS $y=P_{G \otimes D}(y)$, 
then $\norminf{q^{G \otimes D} - y^{[h+1]}} \leq \epsilon$ where $h := 14|G|+3 + \lceil \log (1/\epsilon) +  \log d \rceil$.

 Thus we can compute the probability $q^{G,D}_A = \sum_{t \in F} q^{G \otimes D}_{s_0 A t}$ 
within additive error $\delta > 0$ in time polynomial in the input size: $|G|$, $|D|$ and $\log (1/\delta)$,
in the standard Turing model of computation.
\end{theorem}

We in fact obtain
a much more general result.
For any SCFG, $G$, and corresponding PPS, $x=P_G(x)$, with LFP $q^* > 0$, 
the {\em dependency graph}, 
$H_G = (V,E)$, has  the variables (or the nonterminals of $G$) as nodes 
and has the following edges:
$(x_i, x_j) \in E$ iff 
$x_j$ appears in some monomial in $P_G(x)_i$ with a positive coefficient.
We can decompose the dependency graph $H_G$ into its SCCs, and form the DAG of SCCs, $H'_G$.    
For each SCC, ${\mathcal S}$,
suppose its corresponding equations are $x_S = P_G(x_{\mathcal S},x_{D({\mathcal S})})_S$, where
$D({\mathcal S})$ 
is the set of variables $x_j \not\in \SSS$ such that
there is a path in $H_G$ from some variable $x_i \in \SSS$ to $x_j$.
We call a SCC,  $\SSS$, of $H_G$,  a {\bf\em critical SCC} if the PPS
$x_{\SSS} = P_G(x_{\SSS}, q^{G}_{D(\SSS)})_{\SSS}$ is critical.
In other words, the SCC $\SSS$ is critical if we plug in the LFP values $q^G$ into variables 
that are in lower SCCs, $D(\SSS)$,
then the resulting PPS is critical.
We note that an arbitrary PPS, $x=P_G(x)$ is non-critical if and only if it has no critical SCC.
We define the {\bf\em critical depth}, $\mathfrak{c}(G)$,   of $x=P_G(x)$ 
as follows: it is the
maximum length, $k$, of any sequence ${\mathcal S}_1, {\mathcal S}_2, \ldots,
{\mathcal S}_k$, of SCCs of $H_G$,  
such that for all $i \in \{1, \ldots, k-1\}$, $\SSS_{i+1} \subseteq D(\SSS_i)$,
and furthermore, such that for all $j \in \{1, \ldots, k\}$,
$\SSS_j$ is critical.
Let us call a critical SCC, $\SSS$, of $H_G$ a  {\bf\em bottom-critical SCC}, if $D(\SSS)$ does not
contain any critical SCCs.
By using earlier results (\cite{rmc,EGK10}) 
we can compute 
in P-time the critical SCCs of a PPS, and its critical depth (see the appendix).

PPSs with nested critical SCCs are hard to analyze directly. 
It turns out we can circumvent this by  ``tweaking'' the probabilities in the SCFG $G$ 
to obtain an SCFG $G'$ with no critical SCCs, 
and showing that the ``tweaks'' are small enough 
so that they do not change the probabilities of interest by much.
Concretely:

\begin{theorem}
\label{thm:main-critical-rdnm}
  For any $\epsilon > 0$, and 
for any SCFG, $G$, in SNF form, with $q^G > 0$, 
with critical depth $\cc(G)$,
consider the new SCFG, $G'$, obtained from $G$ by the following process:  
for each bottom-critical SCC, $\SSS$, of $x = P_G(x)$,  find any rule $r = A \xrightarrow{p} B$ of $G$,
such that $A$ and $B$ are both in $\SSS$  (since $G$ is in SNF, such a rule must exist in every critical SCC).  
Reduce the probability $p$, by setting it to\\ $p' = p (1 - 2^{-(14|G|+3)2^{\cc(G)}} \epsilon^{2^{\cc(G)}})$.
Do this for all bottom-critical SCCs.  This defines $G'$, which is non-critical.
Using $G'$ instead of $G$, if we apply R-NM, with parameter $h+2$ to approximate the LFP $q^{G' \otimes D}$ of MPS 
$y=P_{G' \otimes D}(y)$, then $\norminf{q^{G \otimes D} - x^{[h+1]}} \leq \epsilon$ where 
$h := \lceil \log d  + (3 \cdot 2^{\cc(G)}+1)(\log (1/\epsilon) + 14|G|+3) \rceil$.

\noindent Thus we can compute $q^{G,D}_A = \sum_{t \in F} q^{G \otimes D}_{s_0 A t}$ 
within additive error $\delta > 0$ in time polynomial in: $|G|$, $|D|$, $\log (1/\delta)$, and $2^{\cc(G)}$,
in the Turing model of computation.
\end{theorem}

\noindent The proof is very involved, and is in the appendix. There, we also give a family of SCFGs, 
and a $3$-state DFA that checks
the infix probability of string $aa$,
and we explain
why these examples indicate it will likely be difficult to
overcome the exponential dependence on the critical-depth $\cc(G)$ in the above bounds.

\vspace*{-0.12in}

\section{Non-criticality of SCFGs obtained by EM}

\vspace*{-0.1in}

In doing parameter estimation for SCFGs,
in either the supervised or unsupervised (EM) settings (see, e.g., \cite{NedSat08b}),
we are given a CFG, $\HH$, with start nonterminal $S$, and we wish to extend it to an SCFG, $G$, 
by giving probabilities to the rules of $\HH$. 
We also have some probability distribution, $\PP(\pi)$, over the complete derivations, $\pi$, of $\HH$ that start at start non-terminal $S$.
(In the unsupervised case, we begin with an SCFG, and the distribution $\PP$ arises from the prior rule probabilities,
and from the training corpus of strings.) 
We then assign each rule of $\HH$ a (new) probability as follows to obtain (or update) $G$:

\vspace*{-0.08in}

\begin{equation}
  p(A \rightarrow \gamma) := \frac{ \sum_\pi \PP(\pi) C(A \rightarrow \gamma,\pi) }{\sum_\pi \PP(\pi) C(A,\pi)} \label{eq:learn-SCFG}\end{equation}

\vspace*{-0.06in}

\noindent where $C(r,\pi)$ is the number of times the rule $r$ is used 
in the complete derivation $\pi$, and
$C(A,\pi) = \sum_{r \in R_A} C(r,\pi)$. 
Equation (\ref{eq:learn-SCFG}) only makes sense when
the sums $\sum_\pi \PP(\pi) C(A,\pi)$ are finite and nonzero, which we assume; 
we also assume every non-terminal and rule of $\HH$ appears in some complete derivation $\pi$ with $\PP(\pi) > 0$. 

\begin{proposition} If we use parameter estimation to obtain SCFG $G$ using equation (\ref{eq:learn-SCFG}), under the stated assumptions, 
then $G$ is consistent\footnote{Consistency of the obtained SCFGs is well-known; see, e.g., \cite{NedSat06,NedSat08b} \&
references therein; also \cite{SanBen97} has results related to 
Prop. \ref{thm:em-not-critical} for restricted grammars.
}, i.e. $q^G = \vone$, and {\em furthermore}  the PPS $x=P_G(x)$ is non-critical, i.e., $\rho(B_G(\vone)) < 1$.
\label{thm:em-not-critical}  \end{proposition}

It follows from Prop. \ref{thm:em-not-critical} and Thm. \ref{poly-iterations-non-critical},
that for SCFGs obtained by parameter estimation and EM, we can compute the probability 
$q^{G,D}_{A}$ of generating a string in $L(D)$
to within any desired precision in P-time, for any DFA $D$.

\newpage

\appendix

\section{Proof of Theorem \ref{balthm}  (and of Lemma \ref{bal-lem1}).}

\noindent {\bf Theorem \ref{balthm}.}
{\em 
Let $x=P_G(x)$ be any PPS (or MPS), with $n$ variables, associated with a SCFG (or WCFG) $G$, 
and let $y = P_{G \otimes D}(y)$ be the corresponding product MPS,  
for any DFA $D$, with $d$ states.   For any balanced vector
$y \in \BB \subseteq \mathbb{R}^{d^2n}$, with $y \geq 0$,
$\rho(B_{G \otimes D}(y)) = \rho(B_{G}(\mash(y)))$.  
Furthermore, if $\rho(B_{G \otimes D}(y)) < 1$, 
then $\NN_{G \otimes D}(y)$ is defined and balanced, 
$\NN_{G}(\mash(y))$ is defined, and 
$\mash(\NN_{G \otimes D}(y)) = \NN_G(\mash(y))$.
Thus, $\NN_{G \otimes D}$ preserves balance, and
the collapse map $\mash$  ``commutes'' with $\NN$ over non-negative balanced vectors, irrespective of what the DFA $D$ is.}\\

We establish this via a series of lemmas that reveal
many algebraic and analytic properties of balance, collapse, and their interplay
with Newton's method.
Lemma \ref{bal-properties} first establishes a series of algebraic and analytic properties
of arbitrary balanced vectors and matrices.
Lemma \ref{bal-lem1} then uses these to establish 
properties of the specific balanced matrices
and vectors arising during iterations of Newton's method on PPSs (and MPSs), and on 
corresponding product MPSs. 
Theorem \ref{balthm} is an immediate consequence of 
Lemma \ref{bal-lem1}, parts $(i)$\&$(iv)$, below.

\begin{lemma} \label{bal-properties}\mbox{}
Consider the set $\BB \subseteq \real^{d^2n}$ of balanced vectors,
and the set $\BBB \subseteq \real^{d^2n \times d^2n}$ of balanced matrices.
Let $\BB_{\geq 0} = \BB \cap \real^{d^2n}_{\geq 0}$ and $\BBB_{\geq 0} = \BB \cap \real^{d^2n \times d^2n}_{\geq 0}$. 
\begin{itemize}
\item[(i)]
$\BB$ and $\BBB$ are both closed under linear combinations.
In other words:\\ 
$\sum_i \alpha_i v^{\langle i \rangle} \in \BB$ and 
$\sum_i \alpha_i M^{\langle i \rangle} \in \BBB$,
if, $\forall \: i$,  $v^{\langle i \rangle} \in \BB$ and  
$M^{\langle i \rangle} \in \BBB$.   

Furthermore, $\mash$ is a linear map on both $\BB$ and $\BBB$.  In other words:\\
$\mash(\sum_i \alpha_i v^{\langle i \rangle}) = \sum_i \alpha_i \mash(v^{\langle i \rangle} )$ and 
$\mash(\sum_i \alpha_i M^{\langle i \rangle}) = \sum_i \alpha_i \mash(M^{\langle i \rangle})$,\\ whenever, $\forall i$,
$\alpha_i \in \real$,    $v^{\langle i \rangle} \in \BB$,  
and $M^{\langle i \rangle} \in \BBB$.

\item[(ii)] If $M \in \BBB$ and  $v \in \BB$,  then $Mv \in \BB$  and 
$\mash(Mv)= \mash(M)\mash(v)$.

\item[(iii)] If $M, M' \in \BBB$,   then $M M' \in \BBB$ and $\mash(MM')=\mash(M)\mash(M')$.

\item[(iv)] If $M \in \BBB_{\geq 0}$,  and $v \in \real^{d^2n}$ is any vector, then $\mash(Mv) \geq \mash(M)\mash(v)$, where we extend the map $\mash$ 
 to arbitrary $v' \in \real^{d^2n}$ by letting $\mash(v')_A := \min_s \sum_t v'_{(sAt)}$.

\item[(v)] If $M \in \BBB_{\geq 0}$,  then $\rho(M) = \rho(\mash(M))$.    In other words, the collapse operator $\mash$
preserves the spectral radius of balanced non-negative matrices.

\item[(vi)] If $v \in \BB_{\geq 0}$,  then $\norminf{v} \leq  \norminf{\mash(v)}$.
 If $M \in \BBB_{\geq 0}$ then $\norminf{M} \leq d \norminf{\mash(M)}$.
\end{itemize}
\end{lemma}

\begin{proof}\mbox{}\\
\noindent $(i)$: This can be verified directly from the definitions of balance and collapse.
In particular, for any nonterminal $A \in V$,  and any states $s,s' \in Q$:

\begin{eqnarray*}
\sum_t (\sum_i \alpha_i v^{\langle i \rangle})_{(sAt)} & = &   \sum_i \alpha_i \sum_t v^{\langle i \rangle}_{(sAt)}\\
& = &  \sum_i \alpha_i  \mash(v^{\langle i \rangle})_A   \quad \mbox{(because every $v^{\langle i \rangle}$ is balanced)}\\
& = &  \sum_i \alpha_i \sum_t v^{\langle i \rangle}_{(s'At)}\\
& = & \sum_t (\sum_i \alpha_i v^{\langle i \rangle})_{(s'At)}
\end{eqnarray*}

Also, we have $\mash(\sum_i \alpha_i v^{\langle i \rangle})_A :=  
\sum_t (\sum_i \alpha_i v^{\langle i \rangle})_{(sAt)} =  \sum_i \alpha_i  \mash(v^{\langle i \rangle})_A$.

\noindent Likewise,  for any nonterminals $B,C \in V$, and any states $s,u \in Q$
and $v,v' \in Q$:

\begin{eqnarray*}
\sum_{t} (\sum_i \alpha_i M^{\langle i \rangle})_{(sBt), (uCv)} & = &   
\sum_i \alpha_i \sum_{t} M^{\langle i \rangle}_{(sBt), (uCv)}\\
& = &  \sum_i \alpha_i  \sum_{t} M^{\langle i \rangle}_
{(sBt), (uCv')}   \quad \mbox{(because every $M^{\langle i \rangle}$ is balanced)}\\
& = & \sum_t (\sum_i \alpha_i M^{\langle i \rangle})_{(sBt),(uCv')}
\end{eqnarray*}

Similarly, for any nonterminals $B,C$, and any states
$s,v, s',v' \in Q$:

\begin{eqnarray*}
\sum_{t,u} (\sum_i \alpha_i M^{\langle i \rangle})_{(sBt), (uCv)} & = &   
\sum_i \alpha_i \sum_{t,u} M^{\langle i \rangle}_{(sBt), (uCv)}\\
& = &  \sum_i \alpha_i  \sum_{t,u} M^{\langle i \rangle}_
{(s'Bt), (uCv')}   \quad \mbox{(because every $M^{\langle i \rangle}$ is balanced)}\\
& = & \sum_{t,u} (\sum_i \alpha_i M^{\langle i \rangle})_{(s'Bt),(uCv')}
\end{eqnarray*}
\noindent Now, $\mash(\sum_i \alpha_i M^{\langle i \rangle})_{B,C} :=
\sum_{t,u} (\sum_i \alpha_i M^{\langle i \rangle})_{(sBt), (uCv)} = 
\sum_i \alpha_i \sum_{t,u} M^{\langle i \rangle}_{(sBt), (uCv)}
= \sum_i \alpha_i \mash( M^{\langle i \rangle})_{B,C}$.

\noindent $(ii)$:
For any non-terminal $B$ and state $s$:
\begin{eqnarray*} \sum_t (Mv)_{(sBt)} 	& = & \sum_{t,u,C,z} M_{(sBt),(uCz)} v_{uCz}  \\
						& = & \sum_{u,C,z} (\sum_t M_{(sBt),(uCz)})  v_{uCz}  \\
						& = & \sum_{C,u} (\sum_t M_{(sBt),(uCz)}) \sum_z v_{uCz} \quad 
\mbox{(since $M$ is balanced)}\\
						& = & \sum_{C,u} (\sum_t M_{(sBt),(uCz)}) \mash(v)_C \quad \mbox{(since $v$ is balanced)}\\
						& = & \sum_{C} (\sum_{t,u} M_{(sBt),(uCz)}) \mash(v)_C  \\
						& = & \sum_{C} \mash(M)_{B,C} \mash(v)_C \> \mbox{(since $M$ is balanced)}\\\\
						& = & (\mash(M)\mash(v))_B \end{eqnarray*}
which is independent of $s$. So $\mash(Mv)_B =  \sum_t (Mv)_{(sBt)}  = (\mash(M)\mash(v))_B $.

\noindent $(iii)$:
For any non-terminal $D,E$, and states $s,w, x \in Q$:
\begin{eqnarray*} \sum_t (MM')_{(sDt), (wEx)} & = & \sum_{t,u,C,v} M_{(sDt),(uCv)} M'_{(uCv),(wEx)}  \\
						& = & \sum_{u,C,v} (\sum_t M_{(sDt),(uCv)})  M'_{(uCv),(wEx)}   \\
						& = & \sum_{C,u} (\sum_t M_{(sDt),(uCv)}) \sum_v M'_{(uCv),(wEx)}  \quad \mbox{(since $M$ is balanced)}\end{eqnarray*}
Since $M' \in \BBB$, the last sum is independent of $x$, which is what we aimed to show.
Next consider:
\begin{eqnarray*} \sum_{t,w} (MM')_{(sDt), (wEx)} & = & \sum_{t,w,u,C,v} M_{(sDt),(uCv)} M'_{(uCv),(wEx)}  \\
						& = & \sum_{u,w,C,v} (\sum_t M_{(sDt),(uCv)})  M'_{(uCv),(wEx)}   \\
						& = & \sum_{C,u,w} (\sum_t M_{(sDt),(uCv)}) \sum_v M'_{(uCv),(wEx)}  \quad \mbox{(since $M$ is balanced)}\\
						& = & \sum_{C,u} (\sum_t M_{(sDt),(uCv)}) \sum_{v,w} M'_{(uCv),(wEx)} \\
						& = & \sum_{C,u} (\sum_t M_{(sDt),(uCv)}) \mash(M')_{C,E} \quad 
\mbox{(since $M'$ is balanced)}\\
						& = & \sum_{C} \mash(M)_{D,C} \mash(M')_{C,E} \quad \mbox{(since $B$ is balanced)}\\
& = &  (\mash(M) \mash(M'))_{D,E} \end{eqnarray*}
So, $\sum_{t,w} (MM')_{(sDt), (wEx)}$ is independent of $s,x$ and $\mash(MM')_{D,E} = \sum_{t,w} (MM')_{(sDt), (wEx)}
=  (\mash(M) \mash(M'))_{D,E}$, for any $D,E \in V$.\\

\noindent $(iv)$: 
For any non-terminal $B$ and state $s$:
\begin{eqnarray*} \sum_t (Mv)_{(sBt)} 	& = & \sum_{t,u,C,z} M_{(sBt),(uCz)} v_{uCz}  \\
						& = & \sum_{u,C,z} (\sum_t M_{(sBt),(uCz)})  v_{uCz}  \\
						& = & \sum_{C,u} (\sum_t M_{(sBt),(uCz)}) \sum_z v_{uCz} 
\quad \mbox{(since $M$ is balanced)}\\
						& \geq & \sum_{C,u} (\sum_t M_{(sBt),(uCz)}) \min_u \sum_z v_{uCz}  \quad \mbox{(since $(\sum_t M_{(sBt),(uCz)}) \geq 0$) for any  $C,u$}\\
						& = & \sum_C \mash(M)_{B,C} \mash(v)_C  = (\mash(M) \mash(v))_B \end{eqnarray*}

\noindent Since this holds for any $B$ and any $s$,  
$\mash(Mv)_B = \min_s  \sum_t (Mv)_{(sBt)}  \geq  (\mash(M) \mash(v))_B$.\\

\noindent $(vi)$:   (we will prove part $(v)$ below)
Since $v \in \BB_{\geq 0}$,  $v_{(sAt)} \leq \sum_{t'} v_{(sAt')}= \mash(v)_A$ so 
$\norminf{v} \leq \norminf{\mash(v)}$. For $M \in \BBB_{\geq 0}$: 
\begin{eqnarray*} \norminf{M}   & = & \max_{s,B,t} \sum_{u,C,v}  M_{(sBt), (uCv)} \\
& \leq & \max_{s,B} \sum_{u,C,v,t} M_{(sBt), (uCv)} \\
& = & \max_{s,B} \sum_{C,v} \mash(M)_{B,C} \\
& = & \max_{B} d \sum_{C} \mash(M)_{B,C} \\
& =  & d\norminf{\mash(M)} \end{eqnarray*}

\noindent $(v)$: By standard facts from Perron-Frobenius theory 
(see e.g. Theorem 8.3.1 of \cite{HornJohnson85}), the non-negative matrix $\mash(M)$, has as an eigenvalue  
$\rho(\mash(M))$ associated with which is a non-negative eigenvector $v_G \neq 0$.
 That is  $\mash(M)v_G = \rho(\mash(M)) v_G$ for some non-zero $v_G \geq 0$.
Now consider any non-negative balanced vector $u$ with $\mash(u)=v_G$.
(Such a $u$ obviously exists.) 
Let $f(u) = \frac{1}{\rho(\mash(M))} Mu$.
 By part $(ii)$,  $Mu$ is balanced and $\mash(Mu) = \mash(M) v_G =  \rho(\mash(M)) v_G$. 
 So, $f(u)$ is non-negative and balanced and has $\mash(f(u))=v_G$.
 The set of non-negative balanced vector $u$ with $\mash(u)=v_G$ is compact (it is a product of simplices) 
and the continuous function $f$ maps this set into itself. So by 
{\em Brouwer's fixed point theorem}, 
$f$ has a fixed point, that is a $u^*$ with $u^* = \frac{1}{\rho(\mash(M))} Mu^*$. 
 That is, $u^*$ is an eigenvector of $M$ with eigenvalue $\rho(\mash(M))$.
 So $\rho(M) \geq \rho(\mash(M))$.

In the other direction, we use the fact (see, e.g., Theorem 5.6.12 of \cite{HornJohnson85})
that for any square matrix $N$, $lim_{k \rightarrow \infty} \| N^k \|_\infty = 0$ if and only if $\rho(N) < 1$.

Now for $M \in \BBB_{\geq 0}$
assume, for contradiction, that $\rho(M) > \rho(\mash(M))$.
Then $\rho(\frac{1}{\rho(M)} M) = \frac{1}{\rho(M)} \rho(M) = 1 >  \frac{1}{\rho(M)} \rho(\mash(M)) = 
\rho(\frac{1}{\rho(M)} \mash(M))$.
Thus, by the above fact from matrix theory, we have that 
$\lim_{k \rightarrow \infty} \| (\frac{1}{\rho(M)} \mash(M))^k \|_{\infty} = 0$.

But for any $k \geq 1$, 

\begin{eqnarray*}
 0 \leq \| (\frac{1}{\rho(M)} M)^k \|_\infty & \leq &  d \| \mash((\frac{1}{\rho(M)} M)^k) \|_\infty \quad
\mbox{(by part $(vi)$)}\\
& = &  d \| \mash(\frac{1}{\rho(M)} M)^k \| \quad  \mbox{(by part $(iii)$)}\\
& = & d \| (\frac{1}{\rho(M)} \mash(M))^k \|_\infty \quad  \mbox{(by part $(i)$)}
\end{eqnarray*}

And thus, since the right hand side goes to $0$ as $k \rightarrow \infty$, we must also
have  $\lim_{k \rightarrow \infty} \| (\frac{1}{\rho(M)} M)^k \|_\infty = 0$, 
but this is a contradiction, because $\rho(\frac{1}{\rho(M)} M) = 1$.
So, our assumption $\rho(M) > \rho(\mash(M))$ must be false.

Having established both directions, we conclude that $\rho(M) = \rho(\mash(M))$.
\qed
\end{proof}

\noindent {\bf Lemma \ref{bal-lem1}.}
{\em 
Let $\BB_{\geq 0} = \BB \cap \real^{d^2n}_{\geq 0}$ and $\BBB_{\geq 0} = \BB \cap \real^{d^2n \times d^2n}_{\geq 0}$.\\
Let $B_{G}(x)$ denote the Jacobian of the PPS (or MPS) $x=P_G(x)$,
and let $B_{G \otimes D}(y)$ be the Jacobian of MPS $y=P_{G \otimes D}(y)$.\\
Then  $q^{G \otimes D} \in \BB_{\geq 0}$ and $\mash(q^{G \otimes D}) = q^G$, and:
\begin{itemize}
\item[(i)] If $y \in \BB_{\geq 0} \subseteq \real^{d^2n}_{\geq 0}$  
then $B_{G \otimes D}(y) \in \BBB_{\geq 0}$,  
and $\mash(B_{G \otimes D}(y)) = B_G(\mash(y))$.

\item[(ii)] If $y \in \BB_{\geq 0}$, then $P_{G \otimes D}(y) \in \BB_{\geq 0}$, 
and
$\mash(P_{G \otimes D}(y)) = P_G(\mash(y))$.

\item[(iii)] If  $y \in \BB_{\geq 0}$ and 
$\rho(B_G(\mash(y))) < 1$, then $I-B_{G \otimes D}(y)$ is non-singular,\\ 
$(I-B_{G \otimes D}(y))^{-1} \in \BBB_{\geq 0}$, 
and $\mash((I-B_{G \otimes D}(y))^{-1}) = (I - B_G(\mash(y)))^{-1}$.

\item[(iv)] If $y \in \BB_{\geq 0}$ and $\rho(B_G(\mash(y))) < 1$, then 
$\NN_{G \otimes D}(y) \in \BBB$\\ and $\mash(\NN_{G \otimes D}(y)) = \NN_G(\mash(y))$.
\end{itemize}
}

\begin{proof}\mbox{}

Firstly, let us recall why $q^{G \otimes D} \in \BB_{\geq 0}$ 
and $\mash(q^{G \otimes D}) = q^G$.
Recall these are the LFP $q^G$,  of $x = P_G(x)$, and
the LFP $q^{G \otimes D}$ of  $y = P_{G \otimes D}(y)$.
By 
Propositions \ref{prop:lfp-wcfg-is-term}
and \ref{prop:lfp-prod-scfg-is-term}, for any nonterminal $A \in V$,
$q^G_A = \sum_{w \in \Sigma^*} q_A^{G,w}$  is the probability (weight) 
that $G$ generates any finite string $w$.
Likewise $q^{G \otimes D}_{(sAt)} = \sum_{\{w \mid  \Delta^*(s,w) = t \}} q_A^{G,w}$  is the probability (weight) that, starting at $A$,  
$G$ generates a finite string $w$ such that $\Delta^*(s,w) = t$. 
Thus, clearly, for any $A \in V$, and any $s \in Q$,  $q^G_A = \sum_{t \in Q} q^{G \otimes D}_{(sAt)} = \mash(q^{G \otimes D})_A$.
Now we prove the enumerated assertions one by one:\\

\noindent $(i)$: 
We need to argue both that $B_{G \otimes D}(y) \in \BBB_{\geq 0}$,
and that $\mash(B_{G \otimes D}(y)) = B_G(\mash(y))$, for $y \in \BB_{\geq 0}$.
Again, recall that
we are assuming wlog  that $G$ is in SNF form.
We split the proof into cases depending on the type of non-terminal 
$A$ in $B_{G \otimes D}(y)_{(sAt),(uEv)}$.  
Let $\delta_{\alpha,\beta}$ denote the Dirac function: 
$\delta_{\alpha,\beta} := 1$  if $\alpha = \beta$, and $\delta_{\alpha, \beta} :=0$ if $\alpha \neq \beta$.

\noindent {\bf Type {\tt Q}}:  For any non-terminal $A$ of type {\tt Q}, the only
rule in $R_A$ has the form $A \xrightarrow{1} BC$, and
$P_G(x)_A \equiv x_Bx_C$. 
And, for any states $s,t \in Q$, 
$P_{G \otimes D}(y)_{(sAt)} \equiv \sum_{w \in Q} y_{(sBw)} y_{(wCt)}$. 
Thus
$$B_{G \otimes D}(y)_{(sAt),(uEv)} \doteq \frac{\partial P_{G \otimes D}(y)_{(sAt)}}{\partial y_{(uEv)}} =	\delta_{t,v} \cdot \delta_{E,C} \cdot y_{(sBu)} + 
\delta_{s,u} \cdot \delta_{E,B} \cdot y_{(vCt)}  $$
Thus
$$\sum_t B_{G \otimes D}(y)_{(sAt),(uEv)} =  \delta_{E,C} \cdot y_{(sBu)} + \delta_{s,u} \cdot \delta_{E,B} \cdot 
\sum_t y_{(vCt)}$$
Since $y$ is balanced, $ \sum_t y_{(vCt)}$ is independent of $v$, so $ \sum_t B_{(sAt),(uEv)}$ is independent of $v$.  Next we note that:
$$ \sum_{t,u} B_{G \otimes D}(y)_{(sAt),(uEv)} = \delta_{E,C} \sum_u y_{(sBu)} +  \delta_{E,B} \sum_t y_{(vCt)}$$
Thus 
$$\sum_{t,u} B_{G \otimes D}(y)_{(sAt),(uEv)} = \delta_{E,C} \mash(y)_B + \delta_{E,B}  
\mash(y)_C = B_G(\mash(y))$$

\noindent {\bf Type {\tt T}}:
For any non-terminal $A$ of type {\tt T},
$P_{G}(x)_A$ does not depend on $x$, 
and $P_{G \otimes D}(y)_{sAt}$ does not depend on $y$, for any $s,t \in Q$. 
Thus
$ \sum_t B_{G \otimes D}(y)_{(sAt),(uCv)} = 0$, and 
$\sum_{t,u} B_{G \otimes D}(y)_{(sAt),(uCv)} = 0 = B_G(\mash(y))_{A,C}$.\\

\noindent {\bf Type {\tt L}}:
For any non-terminal $A$ of type {\tt L}, recall that 
$P_G(x)_A = \sum_{r \in R_A} p_r x_{B_r}$.
And for any states $s,t$, $ P_{G \otimes D}(y)_{(sAt)} = 
\sum_{r \in R_A} p_r y_{(sB_rt)}$.

Thus, all the entries of $B_G(x))_{A,C}$ and $B_{G \otimes D}(y)_{(sAt),(uCv)}$ are
independent of $x$ and $y$, respectively.
And
$$B_{G \otimes D}(y)_{(sAt),(uCv)} =  
\frac{\partial P_{G \otimes D}(y)_{(sAt)}}{\partial y_{(uCv)}}
= \delta_{s,u} \cdot \delta_{t,v}  \cdot B_G(x)_{A,C}$$ 
Consequently
$\sum_t B_{G \otimes D}(y)_{(sAt),(uCv)} = \delta_{s,u} B_G(x)_{A,C} = $, 
which is independent
of $v$. And, $\sum_{t,u} B_{G \otimes D}(y)_{(sAt),(uCv)} = 
B_G(x)_{A,C}$, which is independent
of $s$ and $v$, and $B_G(x)_{A,C} = B_G(\mash(y))_{A,C}$, 
because $B_G(x)_{A,C}$ is independent of $x$.\\

Having shown that for all nonterminals $A$ and $C$,
and all nonterminals $s, u \in Q$,
the sum  $\sum_t B_{G \otimes D}(y)_{(sAt),(uCv)}$
is independent of $v$.
And we have also shown that for all nonterminals 
$A$ and $C$, the sum 
$\sum_{t,u} B_{G \otimes D}(y)_{(sAt),(uCv)}$ is independent
of $s$ and $v$, and
furthermore, that the latter sum (which is
by definition $\mash(B_{G \otimes D}(y))_{A,C}$), is
equal to $B_G(\mash(y))$.  
Thus our proof for part $(i)$ is complete.\\

\noindent $(ii)$: 
Part $(ii)$ could be proved using a case-by-case 
analysis similar to part $(i)$.
Instead,  we shall use part $(i)$. 
Recall that $P_G(x)$ and $P_{D \otimes G}(y)$ have no polynomials of 
degree more than $2$.  Furthermore:

$$P_G(x) = P_G(0) + B_G(\frac{1}{2}x)x$$
And 
$$P_{G \otimes D}(y) = P_{G \otimes D}(0) + B_{G \otimes D}(\frac{1}{2}y)y$$

By the previous parts of this Lemma, and by 
Lemma \ref{bal-properties},
we know that 
$B_{G \otimes D}(\frac{1}{2}y)y$ is balanced, and 
$\mash(B_{G \otimes D}(\frac{1}{2}y)y) = B_G(\frac{1}{2}\mash(y)) 
\mash(y)$. All that remains is to show that 
$P_{G \otimes D}(0)$ is balanced and that $\mash(P_{G \otimes D}(0)) = P_G(0)$,
and again use the properties established
in Lemma \ref{bal-properties}. 

Now, unless a non-terminal $A$ has type
{\tt T}, $P_G(0)_A = 0$, and for any states $s, t \in Q$, $P_{G \otimes D}(0)_{(sAt)}
=0$.  So, in these cases, there is nothing to prove.
If the nonterminal $A$ does have type {\tt T}, then
$P_G(x)_A = 1$. 
If there is a rule $A \xrightarrow{1} a$, for some $a \in \Sigma$,
then for any state $s \in Q$, there is a unique state 
$t' \in Q$  with $\Delta(s,a) = t'$. 
If instead there is a rule $A \xrightarrow{1} \epsilon$, 
then let $t' := s$. 
In both cases, note that
$\sum_t P_{G \otimes D}(y)_{(sAt)} = 1 = P_G(\mash(y))_A$, since 
$P_{G \otimes D}(y)_{(sAt)} = 1$ when $t=t'$ and $P_{G \otimes D}(y)_{(sAt)} = 0$ 
otherwise.   Thus  also $\mash(P_{G \otimes D}(y)) = P_G(\mash(y))$ in all cases.\\

\noindent $(iii)$: 
By assumption, 
$\rho(B_G(\mash(y))) < 1$, so by Lemma \ref{bal-properties} $(iv)$, 
$\rho(B_{G \otimes D}(y)) < 1$. It is a basic fact that for any square $M \geq 0$
if $\rho(M) < 1$ then $(I-M)$ is non-singular and 
$(I-M)^{-1} = \sum^{\infty}_{i=0}  M^i$.
(See, e.g., \cite{LanTis85}, Theorem 15.2.2, page 531). 
Thus $I - B_{G \otimes D}(y)$ is non-singular, and $(I - B_{G \otimes D}(y))^{-1} =  \sum^{\infty}_{i=0}  (B_{G \otimes D}(y))^i$.
Note that each $(B_{G \otimes D}(y))^i$, for $i \geq 0$, is balanced, by
using the previous parts of this Lemma and Lemma \ref{bal-properties} $(iii)$,
and thus so are the partial sums $\sum^{k}_{i=0}  (B_{G \otimes D}(y))^i$, for any $k \geq 0$.
Therefore $(I - B_{G \otimes D}(y))^{-1} = \lim_{k \rightarrow \infty} \sum_{i=1}^k (B_{G \otimes D}(y_{G \otimes D}))^i$ is a limit 
of balanced non-negative matrices. But then $(I - B_{G \otimes D}(y))^{-1}$ must be balanced,
because the definition of balance for a matrix $M$ requires equalities between continuous (in fact, linear) functions of
the entries, and thus if all the matrices $\sum_{i=1}^k (B_{G \otimes D}(y_{G \otimes D}))^i$ satisfy these
conditions, then so does their limit.

Furthermore $\mash$ is a linear and continuous 
function on matrices, so $\mash((I - B_{G \otimes D}(y))^{-1}) = \sum_{i=1}^\infty \mash(B_{G \otimes D}(y)^i) 
= \sum_{i=1}^\infty \mash(B_{G \otimes D}(y))^i = (I - \mash(B_{G \otimes D}(y)))^{-1}$. 
By part $(i)$ of this Lemma, this is equal to $(I - B_G(\mash(y)))^{-1}$.
Done.\\

\noindent $(iv)$:  By part $(ii)$ of this Lemma, 
$P_{G \otimes D}(y)$ is balanced and 
$\mash(P_{G \otimes D}(y)) = P_G(\mash(y))$. 
Part $(iii)$ of this lemma 
says that 
$(I-B_{G \otimes D}(y))^{-1}$ is balanced and $\mash((I-B_{G \otimes D}(y))^{-1}) = 
(I - \mash(B_{G \otimes D}(y)))^{-1}$. Now we can apply the various algebraic properties
of balanced vectors and matrices from Lemma \ref{bal-properties} to conclude that
$$\NN_{G \otimes D}(y) := y + (I - B_{G \otimes D}(y)^{-1} (P_{G \otimes D}(y) - y)$$
is balanced and that
$\mash(\NN_{G \otimes D}(y)) = \mash(y) + 
(I - B_G(\mash(y)))^{-1}(P_G(\mash(y)) - \mash(y)) = \NN_G(\mash(y))$.
\qed
\end{proof}

As mentioned already, Theorem \ref{balthm} follows immediately from 
Lemma \ref{bal-lem1}, parts $(i)$\&$(iv)$.

\section{Proofs for Section 4}

We will first show how to compute in P-time
the critical SCCs and the critical depth of a PPS.  We then 
proceed to prove the main theorems of the section:
Theorems \ref{poly-iterations-non-critical} and \ref{thm:main-critical-rdnm}.

Let $x=P(x)$ be a PPS (wlog in SNF), with LFP $q^* >0$, 
 let $B(x)$ be its Jacobean matrix,
and let $H=(V,E)$ be its dependency graph.
If $B$ is a square matrix and $I,J$ are subsets of indices,
we will use $B_{I,J}$ to denote the submatrix with rows in $I$
and columns in $J$, and we use $B_I$ to denote the square submatrix $B_{I,I}$.

\begin{proposition} 
\label{thm:p-time-critical}
Given a PPS $x=P(x)$ with LFP $q^* >0$, we can compute in
polynomial time its critical SCCs and its critical depth.
\end{proposition}
\begin{proof}
We know that for each SCC $\SSS$ of $H$, either all the
variables (nodes) of the SCC have value 1 in the LFP $q^*$,
or they all have value $<1$; moreover, if they have value 1,
then so do all the variables that they can reach in $H$,
i.e., $q^*_{\SSS}=\vone$ implies $q^*_{D(\SSS)}=\vone$ \cite{rmc}.
Furthermore, we can determine which variables and SCCs have
value $1$, and which value $<1$, in polynomial time \cite{rmc}
(this was improved to strongly polynomial time in \cite{EGK10}).
We also know that $\rho(B(q^*)) \leq 1 $, thus a PPS is critical
iff  $\rho(B(q^*)) = 1 $. 
Furthermore, by Theorem 3.6 of \cite{ESY12}, if $q^* <\vone$, then
$\rho(B(q^*)) < 1 $. 

Therefore, for each SCC $\SSS$, we can determine whether
it is critical as follows. 
If  $q^*_{\SSS}<\vone$ then $\SSS$  is not critical.
If $q^*_{\SSS}=\vone$, then $\SSS$ is critical iff $\rho(B(\vone)_{\SSS}) =1$,
and it is not critical iff  $\rho(B(\vone)_{\SSS}) < 1 $;
we can determine which of the two is the case as follows. 
Since the spectral radius of $B(\vone)_{\SSS}$ is at most 1,
$\rho(B(\vone)_{\SSS}) = 1$  iff there is a vector
$u \neq 0$ such that $(B(\vone)_{\SSS})\cdot  u = u$
(and we can take $u \geq 0$ to be an eigenvector
for the eigenvalue 1 in this case since the matrix is nonnegative),
or equivalently since the constraints are homogeneous in $u$, 
this is the case iff the set of linear equations
$\{ (B(\vone)_{\SSS}) \cdot u  = u; \sum_i u_i =1 \}$ has a solution.
This can be checked in (strongly) polynomial time by standard methods.

Once we have identified the critical SCCs, it is straightforward
to compute the critical depth in linear time in the
size of the DAG of SCCs by a traversal of the DAG in
topological order.
\qed
\end{proof}

\begin{proposition} 
\label{thm:critical-PPS-has-critical-SCC}
A PPS $x=P(x)$ is critical if and only if at least one
of its SCCs is critical.
\end{proposition}
\begin{proof}
(Only if): Suppose first that the PPS is critical,
i.e., that $\rho(B(q^*))=1$.
Let $v \geq 0$,  $v \neq 0$, be an eigenvector of $B(q^*)$ for the eigenvalue 1,
i.e., $B(q^*)v=v$.
Let $\SSS$ be a lowest SCC that contains a variable with nonzero value
in $v$, i.e. $v_{\SSS} \neq 0$ and $v_{D(\SSS)}=0$. 
Then  $v_{\SSS} = B(q^*)_{\SSS,\SSS \cup D(\SSS)} \cdot v_{\SSS \cup D(\SSS)} =
B(q^*)_{\SSS} \cdot v_{\SSS}$. Thus, $v_{\SSS}$ is
an eigenvector of $B(q^*)_{\SSS}$ with eigenvalue 1, hence
$\rho(B(q^*)_{\SSS}) \geq 1$,
and since we always have $\rho(B(q^*)_{\SSS}) \leq 1$, 
if follows that  $\SSS$ is a critical SCC.

(If): Conversely, suppose that there is a critical SCC, and let $\SSS$
be a highest critical SCC in the DAG of SCC's.
Then $\rho(B(q^*)_{\SSS}) = 1$. Let $u \geq 0$ be an eigenvector 
of $B(q^*)_{\SSS}$ with eigenvalue 1.
Let $E(\SSS)$ be the (possibly empty) set of variables which depend on variables in $\SSS$ but are not themselves in $\SSS$. 
If $E(\SSS) = \emptyset$ then let $v$ be a vector with
$v_{\SSS}=u$ and $v_i=0$ for all variables $x_i \notin \SSS$.
Then $B(q^*)v=v$, i.e., $v$ is an eigenvector of $B(q^*)$ with eigenvalue 1,
hence $\rho(B(q^*)) \geq 1$ and the PPS is critical.

Suppose that $E(\SSS)$ is nonempty.
Then $E(\SSS)$ contains no critical SCCs by our choice of $\SSS$. 
This implies by our proof above for the (only if) direction that the PPS $x_{E(\SSS)} = P(x_{E(\SSS)},x_{D(E(\SSS))})$
is not critical, 
i.e., $\rho(B(q^*)_{E(\SSS)}) < 1$.
Thus, $(I - B(q^*)_{E(\SSS)})^{-1}$ exists. 
Let $v$ be the vector with $v_{\SSS}=u$, 
$v_{E(\SSS)} = (I - B(q^*)_{E(\SSS)})^{-1} B(q^*)_{E(\SSS), \SSS} \cdot v_{\SSS}$ and $v_i=0$ for all $x_i$ not in either $\SSS$ or $E(\SSS)$. 

We claim that $B(q^*)v = v$.
If $x_i$ does not depend on a variable in $\SSS$, then any $x_j$ which 
$x_i$ depends on also does not depend on $\SSS$ and so has $v_j = 0$. So $(B(q^*)v)_i=0=v_i$. Next we consider $(B(q^*)v)_{\SSS}$. Since $D(\SSS)$ is disjoint from $\SSS$ and $E(\SSS)$, $v_{D(\SSS)} = 0$. So $(B(q^*)v)_{\SSS} = (B(q^*))_{\SSS} \cdot v_{\SSS} = v_{\SSS}$. Lastly consider $(B(q^*)v)_{E(\SSS)}$. 
\begin{eqnarray*} (B(q^*)v)_{E(\SSS)} & = & B(q^*)_{E(\SSS)} \cdot v_{E(\SSS)} + B(q^*)_{E(\SSS), \SSS} \cdot v_{\SSS} \\
& = & v_{E(\SSS)} - (I - B(q^*)_{E(\SSS)}) \cdot  v_{E(\SSS)} + B(q^*)_{E(\SSS), \SSS} \cdot v_{\SSS}  \\
& = & v_{E(\SSS)} - B(q^*)_{E(\SSS), \SSS} \cdot v_{\SSS} + B(q^*)_{E(\SSS), \SSS} \cdot v_{\SSS}  \\
& = & v_{E(\SSS)}  \end{eqnarray*}
So $B(q^*)v=v$. 
Therefore, $\rho(B(q^*) \geq 1$ and hence the PPS is critical.
\qed
\end{proof}

\medskip

In the remainder of this section we will prove 
Theorem \ref{thm:main-critical-rdnm}, and along the
way, we will also  establish Theorem \ref{poly-iterations-non-critical}.
The proof of Theorem \ref{thm:main-critical-rdnm} is long and involved.
We first need to recall, and 
establish, a series of Lemmas and Theorems.

\begin{lemma}{(Lemma C.3 of \cite{bmdp})} \label{quantpf} If $A$ is a non-negative matrix, 
and vector $u > 0$ is such that $Au \leq u$ and $\|u\|_\infty \leq 1$,   
and $\alpha , \beta \in (0,1)$ are constants such that for every $i \in \{1,...n\}$, 
one of the following two conditions holds:
\begin{itemize}
\item[(I)]   $(Au)_i \leq (1 - \beta) u_i$

\item[(II)]  there is some $k$, $1 \leq k \leq n$, and  
some $j$, such that $(A^k)_{ij} \geq \alpha$ and $(Au)_j \leq (1 - \beta) u_j$.
\end{itemize} 
then $(I-A)$ is non-singular, $\rho(A) < 1$,\footnote{Although the fact that the conditions imply also 
that $\rho(A) < 1$ is not stated
explicitly in Lemma C.3 of \cite{bmdp}, it is indeed established in the proof in \cite{bmdp}.}
and
$$\norminf{(I-A)^{-1}} \leq \frac{n}{u_{\min}^2\alpha \beta}$$ \end{lemma}

\begin{lemma}{(Lemma A.4 of \cite{ESY12})}
\label{invdet} Let A be a non-singular $n \times n$ matrix with rational entries. If the product of the denominators of all these entries is $m$, then
$$\|A^{-1}\|_\infty \leq nm\|A\|_\infty^n$$\end{lemma}

 \begin{lemma}{(Lemma 5.4 from \cite{lfppoly}; or see Lemma 3.7 from \cite{ESY12})} 
\label{cone} 
Let $x=P(x)$ be a monotone system of polynomial equations which has a LFP $q^*$. 
For any positive vector 
$\mathbf{d} \in \mathbb{R}^n_{> 0}$ 
that satisfies $B(q^*) \mathbf{d} \leq \mathbf{d}$, any positive real value $\lambda > 0$, 
and any nonnegative vector $x \in \mathbb{R}^n_{\geq 0}$, 
if $q^* - x \leq \lambda \mathbf{d}$, and $(I-B(x))^{-1}$
exists and is nonnegative, then 
$$q^* - \NN(x) \leq \frac{\lambda}{2} \mathbf{d}$$\end{lemma}

\begin{theorem}{(Theorem 3.12 of \cite{ESY12})}  
\label{thm:1comp-stoc} 
For a PPS, $x=P(x)$ in $n$ variables, in SNF form, with LFP $q^*$, such that 
$0 < q^* < 1$, 
for all $i= 1, \ldots, n$: $ 1-q^*_i \geq 2^{-4|P|}$.
In other words, $\| q^* \|_{\infty} \leq 1-2^{-4|P|}$. 
\end{theorem}

\begin{theorem}{(Theorem 4.6 of \cite{bmdp})}  \mbox{}
\label{normbound}

(i) if  $x=P(x)$ is a PPS with $q^* < \vone$ and $0 \leq y < \vone$ then 
$$\norminf{(I-B(\frac{1}{2}(y+q^*)))^{-1}} \leq  2^{10|P|} \max \{2(\vone-y)_{\min}^{-1}, 2^{|P|}\}$$

(ii) if  $x=P(x)$ is a strongly connected PPS with $q^* = \vone$ and $0 \leq y < \vone$, then
$$\norminf{(I -B(y))^{-1}} \leq 2^{4|P|} \frac{1}{(\vone-y)_{\min}}$$
\end{theorem}

\begin{lemma} 
\label{lem:bound-on-cone-up-down}
If $x=P(x)$ is
a strongly connected PPS (in SNF form), with Jacobian $B(x)$,
and if $B(\frac{1}{2} \vone)v \leq v$ for some vector $v > 0$, then $\frac{\norminf{v}}{v_{\min}} \leq 2^{|P|}$ \end{lemma}

\begin{proof} (This proof is  a variant of that of Lemma 3.10 in \cite{ESY12}.)
Let $l =  \arg \max_i v_i$, and let $k = \arg \min_j v_j$.
Since $x=P(x)$ is in SNF form,
every non-zero entry of the matrix
$B(\frac{1}{2}\textbf{1})$ is
either $1/2$ or is a coefficient
of some monomial in some polynomial $P(x)_i$ of $P(x)$.
Moreover, $B(\frac{1}{2}\textbf{1})$ is irreducible.
Calling the entries of $B(\frac{1}{2}\textbf{1})$, $b_{i,j}$, we have a sequence
of {\em distinct} indices, $i_1,i_2,\ldots, i_m$, with $l = i_1$, $k = i_m$, $m \leq n$, where
 each $b_{i_j i_{j+1}} > 0$.
(Just take the ``shortest positive path'' 
 from $l$ to $k$.)
For any $j$:

\vspace*{-0.06in}

$$(B(\frac{1}{2}\textbf{1})v)_{i_{j+1}} \geq b_{i_ji_{j+1}}  v_{i_j}$$

By simple induction:
$v_k \geq (\prod_{j=1}^{m-1} b_{i_ji_{j+1}}) v_l$.
Note that  $|P|$ includes the encoding size of
each positive coefficient of every polynomial $P(x)_i$. 
We argued before that each $b_{i_ji_{j+1}}$ is either a coefficient of $x=P(x)$,
or is equal to $1/2$.  Furthermore, if we consider the equation $x_{i_j} = P(x)_{i_j}$,
and denote its encoding size as $|P_{i_j}|$, then it is easy to see $b_{i_j i_{j+1}} \geq 2^{-|P_{i_j}|}$,
because either $b_{i_j i_{j+1}}$ appears in $P(x)_{i_j}$, or else $b_{i_j i_{j+1}} = 1/2$, but
it is always the case that $|P_{i_j}| \geq 1$.  
Now, the $i_j$'s are distinct  (because we are using a shortest path).
Therefore, since $|P| = \sum^n_{i=1} |P_i|$, 
we must have $\prod_{j=1}^{m-1} b_{i_ji_{j+1}} \geq 2^{-|P|}$, 
and thus we have: $v_k \geq 2^{-|P|} v_l$.
\qed
\end{proof}

\begin{theorem}  \label{normtolconv} If $x=P(x)$ is an MPS 
with $n$ variables, with 
LFP $q^* \leq 1$, and $\rho(B(q^*)) < 1$, and 
if we use {\em any} rounded-down Newton iteration method
defined by $x^{[0]} := 0$, and for all $k \geq 0$,
and $x^{[k+1]} :=  \max (0, \NN(x^{(k)}) - e_{k})$, 
where  $e_{k}$ is {\em some error vector} such 
that $0 \leq (e_k)_i \leq 2^{-(h+2)}$ for all $i \in \{1,\ldots,n\}$, 
then for any $0 < \epsilon \leq 1$,
$\norminf{q^*-x^{[h + 1]}} \leq \epsilon$, whenever the chosen parameter $h$ satisfies 
$h \geq \lceil \log  \norminf{(I - B(q^*))^{-1}} + \log \frac{1}{\epsilon} \rceil$.
 \end{theorem}
 \begin{proof} 
We shall use Lemma \ref{cone} to prove this.
We need to find a vector $v$, with $B(q^*)v \leq v$ and $v > 0$,
called a {\em cone vector}, 
such that we can bound the ratio $\frac{v_{\max}}{v_{\min}}$. 
Here $v_{\max} = \max_i v_i$, and $v_{\min} = \min_i v_i$.

Since we know that $\rho(B(q^*)) < 1$, 
we have that $(I- B(q^*))$ is nonsingular, and $(I- B(q^*))^{-1} = \sum_{i=0}^\infty B(q^*)^i$. We simply take $v :=\frac{1}{\norminf{(I - B(q^*))^{-1}}} (I-B(q^*))^{-1}\vone$ as our cone vector. 

Then $B(q^*)v = v - \frac{1}{\norminf{(I - B(q^*))^{-1}}} \vone \leq v$ and $v = \frac{1}{\norminf{(I - B(q^*))^{-1}}}(\vone + B(q^*)\vone + B(q^*)^2 \vone ...) \geq \frac{1}{\norminf{(I - B(q^*))^{-1}}} \vone$. The latter not only shows 
that $v > 0$, but also that $v_{\min} \geq \frac{1}{\norminf{(I - B(q^*))^{-1}}}$. Recall that by definition, since $(I - B(q^*))^{-1}$ is
non-negative, $\norminf{(I-B(q^*))^{-1}}$
is the maximum row 
sum of any row of $(I - B(q^*))^{-1} = \sum^{\infty}_{i=0} B(q^*)^i$. 
It follows that $v_{\max} \leq 1$, since $B(q^*)^0 = I$.
 
Now, $x^{[0]} := 0$, and $q^* \leq 1$, so we 
know that $q^* - x^{[0]} \leq 1 \leq \norminf{(I - B(q^*))^{-1}} v \leq 2^{h} \epsilon v$
(by definition of $h$).
Now, for all $k > 0$,  $e_k \leq 2^{-(h+2)} \vone \leq \frac{1}{4}\epsilon \frac{1}{\norminf{(I - B(q^*))^{-1}}} \vone \leq \frac{1}{4}\epsilon v$. 

Applying Lemma \ref{cone}, 
if $q^* - x^{[k]} \leq \lambda v$, then $q^* - x^{[k+1]} \leq q^* - \NN(x^{[k]}) + e_{k} \leq 
(\frac{\lambda}{2} + \frac{1}{4})\epsilon v$.
It follows by induction that, for all $k \geq 1$,  $q^* - x^{[k]} \leq (2^{h-k} + \frac{1}{2}) \epsilon v$
When $k=h+1$, this gives $q^*-x^{[h + 1]} \leq \epsilon v$. Since $v_{\max} = \norminf{v} \leq 1$, 
this means that $\norminf{q^*-x^{[h + 1]}} \leq \epsilon$ as required. \qed \end{proof}

\begin{theorem}  \label{newtonnocritical} If the PPS $x=P(x)$ with LFP solution $q^*$ has $\rho(B(q^*)) < 1$ and we 
use any rounded-down Newton iteration, 
starting at $x^{[0]} = 0$, defined by
$x^{[k+1]} =  \max (0, x^{[k]} + (I-B(x^{[k]}))^{-1}(P(x^{[k]}) - x^{[k]}) - e_k)$, for any error vectors $e_{k}$
where 
$0 \leq (e_{k})_i \leq 2^{-(h + 2)}$ for all $i \in \{1,\ldots, n\}$, 
then for any given $0 < \epsilon \leq 1$,
$\norminf{q^*-x^{[h + 1]}} \leq \epsilon$, where $h =  14|P|+3 + \lceil \log (1/\epsilon) \rceil$. 

\end{theorem}

\noindent Theorem \ref{newtonnocritical} follows from
Theorem \ref{normtolconv} and an upper bound on 
$\norminf{(I-B(q^*))^{-1}}$. The following Lemma gives us this, from
which Theorem \ref{newtonnocritical} follows immediately:
\begin{lemma} \label{nocriticalnorm} If the PPS $x=P(x)$ with LFP solution $q^*$ has $\rho(B(q^*)) < 1$ then
$$\norminf{(I-B(q^*))^{-1}} \leq  2^{14|P|+3} $$
\end{lemma}
\begin{proof}
We split into several cases, based on $q^*$.

\noindent {\em Case 1}: $q^* < {\mathbf 1}$. 
In this case we just need to use Theorem \ref{normbound} $(i)$,
in which we set $y := q^*$,  combined with Theorem
\ref{thm:1comp-stoc},  to conclude that:
$$\norminf{(I-B(q^*))^{-1}} \leq  2^{14|P|+1} $$

\noindent {\em Case 2}: $q^* = {\mathbf 1}$.
In this case we can instead use the following result from \cite{ESY12}:

\begin{lemma} \label{B1normbound} For a PPS $x=P(x)$, if $(I-B(\vone))$ is non-singular then
 $$\norminf{(I-B(\vone))^{-1}} \leq 3^nn2^{|P|} \leq 2^{3|P|}$$ \end{lemma}

 \begin{proof} The proof of this is basically 
identical to a proof 
in \cite{ESY12} for a closely related fact, which was based on more assumptions
(but not all of the assumptions were needed).

If we take $(I - B(\textbf{1}))$ to be the matrix $A$ of Lemma \ref{invdet}, then
noting that the product of all the denominators in $(I-B(\textbf{1}))$ is at most $2^{|P|}$,
this  yields:
$$\|(I-B(\textbf{1}))^{-1}\|_\infty \leq n2^{|P|}\|(I-B(\textbf{1}))\|_\infty^n$$
Of course $\|(I-B(\textbf{1}))\|_\infty \leq 1 + \|B(\textbf{1})\|_\infty \leq 3$ \  (note 
that here
we are using the fact that the system is in SNF normal form).  Thus
$$\|(I-B(\textbf{1}))^{-1}\|_\infty \leq 3^nn2^{|P|}$$

\noindent Furthermore, as discussed in \cite{ESY12}
(see section A.6, first paragraph),  
for any PPS $x=P(x)$ we can assume 
wlog that the equation for every variable
requires at least 3 bits, and thus
that $|P| \geq  3 n \geq  n \log 3 + \log n$.
Therefore  $3^nn2^{|P|} \leq 2^{3|P|}$.\qed
\end{proof}

\noindent {\em Case 3}: Neither $q^* < {\mathbf 1}$ nor  $q^* = {\mathbf 1}$.
To finish the proof of Lemma  \ref{nocriticalnorm},
we will combine the above two results for the first two cases to deal with the case 
when neither $q^* < {\mathbf 1}$ nor  
$q^* = {\mathbf 1}$,
but that nevertheless $\rho(B(q^*)) < 1$.  (It is indeed possible for all
three of these conditions to hold, when some coordinates of $q^*$ are $1$, and others less than $1$.)

Let $A$ (for ``always'') denote the set of variables $x_i$ for which $q^*_i = 1$, 
and let $M$ (for ``maybe'') denote the set of variables $x_i$ for which $0 < q^*_i < 1$. 
We can obviously assume that both $A$ and $M$ are non-empty; otherwise one of the two above theorems gives the result. 
Furthermore, variables in $A$ obviously cannot depend on those in $M$
(neither directly nor indirectly). Thus we can describe $B(q^*)$ 
by the following block decomposition 
$$B(q^*) = \begin{pmatrix} B(q^*)_M &  B(q^*)_{M,A} \\ 0 & B(q^*)_A \end{pmatrix}$$ 
We need a lemma:
 
\begin{lemma}
\label{lem:block-inverse-norm-bound}
For any matrix $M$ satisfying the block decomposition
given by \\
$M = \begin{pmatrix} A & B \\ 0 & D \end{pmatrix}$, 
if both $A$ and $D$ are square and non-singular matrices,
then $M$ is also non-singular, and:

$$\norminf{M^{-1}}  \leq \max \{ \norminf{A^{-1}} + \norminf{A^{-1}} \norminf{B} \norminf{D^{-1}}, 
\norminf{D^{-1}}\}$$

\end{lemma} 
\begin{proof}
The standard formula for the blockwise inverse of a matrix
gives\\ 
$\begin{pmatrix} A & B \\ 0 & D \end{pmatrix}^{-1} = \begin{pmatrix} A^{-1} & \ - A^{-1}BD^{-1} \\ 0 & D^{-1} \end{pmatrix}$, provided that $A$ and $D$ are non-singular. 
(The formula can easily be verified directly by multiplying by $\begin{pmatrix} A & B \\ 0 & D \end{pmatrix}$ .)

Now recall that the $l_\infty$ norm for a matrix $C$ is $\norminf{C} =
\max_i \sum_j \abs{C_{ij}}$, i.e., it is the maximum sum across any row of
the absolute value of the entries of the row.
So
\begin{eqnarray*}
\norminf{M^{-1}} & \leq & \max \{ \norminf{A^{-1}} + \norminf{A^{-1}} \norminf{B} \norminf{D^{-1}}, 
\norminf{D^{-1}}\}
\end{eqnarray*}
\qed
\end{proof}

Now,
$(I - B(q^*)) = \begin{pmatrix} I - B(q^*)_M & \ -B(q^*)_{M,A} \\ 0 & I-B(q^*)_A\end{pmatrix}$, so 
$\norminf{(I - B(q^*))^{-1}} \leq \max \{ \norminf{(I - B(q^*)_M)^{-1}} + \norminf{(I - B(q^*)_M)^{-1}} \norminf{B(q^*)_{M,A}} \norminf{(I-B(q^*)_A)^{-1}}, \newline \norminf{(I-B(q^*)_A)^{-1}} \}$. \\
Since we always wlog assume that $x=P(x)$ is a PPS is SNF normal form, 
$\norminf{B(q^*)} \leq 2$. More specifically, $\norminf{ B(q^*)_{M,A}} \leq 2$.
By 
Case 1, since ${\mathbf 0} < q^*_M < {\mathbf 1}$, 
$\norminf{(I-B(q^*)_M)^{-1}} \leq  2^{14|P_M|+1}$, 
where $|P_M|$ denotes the encoding size of the system of equations $x_M=P(x_M,1_A)_M$,
restricted to the variables in $M$, and with $1$ plugged in for all variables in $A$.
Also, 
by Lemma \ref{B1normbound}, since $q^*_A = {\mathbf 1}$,
$\norminf{(I-B(q^*)_A)^{-1}} \leq 2^{3|P_A|}$,
where $x_A = P(x)_A$ denotes the system of equations restricted to variables
in $A$ (note that these do not depend on variables in $M$).
Thus,
$$  \norminf{(I - B(q^*))^{-1}} \leq \max \{ 2^{14|P_M|+1} + 2^{14|P_M|+2 + 3|P_A|}, 2^{3|P_A|} \}$$
This can be simplified to $\norminf{(I - B(q^*))^{-1}} \leq 2^{14|P| +3}$.
This completes the proof of
Lemma \ref{nocriticalnorm}.
\qed
\end{proof}

We now have enough to deal with the non-critical case of Theorem \ref{poly-iterations-non-critical}.\\

\noindent {\bf Theorem  
\ref{poly-iterations-non-critical}.}
{\em  
For any $\epsilon > 0$, and for an SCFG, $G$, if the PPS $x= P_G(x)$  has LFP $0 < q^G \leq 1$ and $\rho(B_G(q^G)) < 1$, 
then if we use R-NM with parameter $h+2$ to approximate the 
LFP solution of the MPS $y=P_{G \otimes D}(y)$, 
then $\norminf{q^{G \otimes D} - y^{[h+1]}} \leq \epsilon$ where $h := 14|G|+3 + \lceil \log (1/\epsilon) +  \log d \rceil$.

 Thus we can compute the probability $q^{G,D}_A = \sum_{t \in F} q^{G \otimes D}_{s_0 A t}$ 
within additive error $\delta > 0$ in time polynomial in the input size: $|G|$, $|D|$ and $\log (1/\delta)$,
in the standard Turing model of computation.}

\begin{proof}
Lemma \ref{bal-lem1} yields that $(I-B_{G \otimes D}(q^{G \otimes D}))^{-1} \in \BBB_{\geq 0}$, and that
$\mash((I-B_{G \otimes D}(q^{G \otimes D}))^{-1}) = (I - (B_{G}(q^G))^{-1}$. 
Lemma \ref{bal-properties}$(vi)$ 
relates the norms:\\ $\norminf{(I-B_{G \otimes D}(q^{G \otimes D}))^{-1}} \leq d \norminf{(I - (B_{G}(q^G))^{-1}}$. 
We need a bound on the latter norm. 
Lemma \ref{nocriticalnorm}  
shows $\norminf{(I-B_G(q^G))^{-1}} \leq  2^{14|G|+3}$. So\\ 
$\norminf{(I-B_{G \otimes D}(q^{G \otimes D}))^{-1}} \leq  d 2^{14|G|+3}$. 
Plugging this bound into Theorem \ref{normtolconv} 
yields the result. \qed \end{proof}

To deal with critical SCCs, we need a way to analyse how an error in 
the LFP $q^*$ inside one SCC, ${\mathcal S}$,  where $q^*_{\mathcal S}=1$, affects those SCCs that depend on it:

\begin{theorem} \label{errorup} Given  a PPS, $y = P(y)$ in SNF form, such that for a subvector $x$
of $y$, whose equations are $x=P(x,y_{D(x)})$, when restricting $y=P(y)$
to the variables in $x$, and if  we let $y_{D(x)} := z$,
for a real-valued vector
$0 \leq z < \vone$, and if the resulting PPS, $x= P(x,z)$ has
LFP $q^*_z>0$, and if $q^*_\vone$ is the LFP solution of $x=P(x,\vone)$
(note that $q^*_\vone \geq q^*_z$),
then: \\
\begin{itemize}
\item[(i)] If $q^*_\vone < \vone$ then, $\norminf{q^*_{\vone} - q^*_z} \leq 2^{14|P|+2} \norminf{\vone - z}$\\

\item[(ii)] If the PPS $x=P(x,\vone)$ is strongly connected and $q^*_\vone=\vone$ then $\norminf{\vone - q^*_z} \leq 2^{3|P|}\sqrt{\norminf{\vone-z}}$\\
\item[(iii)]
If the PPS, $x=P(x,1)$, is strongly connected and $q^*_\vone=\vone$, and $\rho(B(\vone,\vone)) < 1$ 
then $\norminf{\vone - q^*_z} \leq 2^{3|P|} \norminf{\vone - z}$ 
\end{itemize}\end{theorem}

Bad examples given in \cite{lfppoly} (see also
\cite{ESY13}), show that there are critical PPSs with $q^*_1 = 1$, and
 with $\norminf{\vone - q^*_z} \geq \sqrt{\norminf{\vone-z}}$. Thus we cannot hope to get a bound linear 
in $\norminf{\vone-z}$ in all cases. Cases $(i)$ and $(iii)$ 
of Theorem \ref{errorup} say that we can get a linear bound {\em except for} critical PPSs, 
where we indeed need a square root in the strongly connected case (case $(ii)$).

\begin{proof}[of Theorem \ref{errorup}]
We first prove the following:

\begin{lemma} 
\label{lem:same-as-in-sey13}
For  $0 \leq z \leq z' \leq \vone$, and
for
all $0 \leq x \leq 1$, $\norminf{P(x,z') - P(x,z)} \leq 2 \norminf{z - z'}$ \end{lemma}

\begin{proof}
Consider the $k$'th coordinate, $P(x,y)_k$,
of the PPS polynomials  $P(x,y)$, in SNF form.
We distinguish cases based on the type of $x_k$. 
If $x_k$ has type {\tt Q}: 
then $P(x,z)_k$ and $P(x,z')_k$ both have
the form $x_i x_j$, or both have form $z^{(')}_i x_j$, or both the form $x_i z^{(')}_j$, or both 
the form $z^{(')}_i z^{(')}_j$.
Thus, since $0 \leq z \leq z' \leq 1$,  and $0 \leq x \leq 1$,  we have 
$0 \leq P(x,z')_k - P(x,z)_k \leq z'_i z'_j - z_i z_j \leq 2  \norminf{z - z'}$.

In the case where $x_k$ has type {\tt L}, we have $0 \leq P(x,z')_k - P(x,z)_k \leq 
\sum_j p_{k,j} (z'_j - z_j) \leq  \norminf{z - z'}$,
because the coefficients $p_{k,j}$ of the type {\tt L} equation must sum to $\leq 1$.

Finally, if $x_k$ has type {\tt T}, $P(x,z)_k$ and $P(x,z')_k$ are equal constants,
so their difference is $0$.
\qed
\end{proof}

\begin{lemma} 
\label{neednormPPS} 
If $x = P(x,z)$ 
is a PPS with LFP $q^*_z > 0$ and $x=P(x,z')$ has LFP $q^*_{z'} > 0$ for some $0 \leq z \leq z' \leq \vone$, and $(I - B(\frac{1}{2}(q^*_{z'}+q^*_z),z'))$ is non-singular then 
$$\norminf{q^*_{z'} - q^*_z} \leq 2 \norminf{(I - B(\frac{1}{2}(q^*_{z'}+q^*_z),z'))^{-1}} \norminf{ z'- z}$$
\end{lemma}
 \begin{proof}
From Lemma 4.3 of \cite{bmdp}, applied to the PPS $x=P(x,z')$, 
(where we let $y := q^*_z$), we have:
$$(q^*_{z'} - q^*_{z})= (I - B(\frac{1}{2}(q^*_{z'}+q^*_z),z'))^{-1}(P(q^*_{z},z') - q^*_{z})$$
We can take norms:
$$\norminf{q^*_{z'} - q^*_{z}}= \norminf{(I - B(\frac{1}{2}(q^*_{z'}+q^*_z),z'))^{-1}} 
\norminf{ (P(q^*_{z},z') - q^*_{z})}$$ 
Now we just apply
Lemma \ref{lem:same-as-in-sey13},
to obtain that $\norminf{ (P(q^*_{z},z') - q^*_{z})} \leq 2 \norminf{ z'- z}$.
\qed
\end{proof}

To get parts $(i)$ and $(ii)$ of Theorem \ref{errorup}, we apply Theorem \ref{normbound}.
For establishing $(i)$ of Theorem \ref{errorup}, we need to apply $(i)$ of Theorem 
\ref{normbound} to the PPS, $x=P(x,\vone)$, with $y:=q^*_z$. 
This gives 
$$\norminf{(I-B(\frac{1}{2}(q^*_z+q^*_\vone), \vone))^{-1}} \leq  2^{10|P|} \max \{2(\vone-q^*_z)_{\min}^{-1}, 2^{|P|}\}$$
 Now, since in part $(i)$ 
of Theorem \ref{errorup}, we are given that $q^*_\vone < 1$,
we know that
 $q^*_z \leq q^*_\vone \leq {\mathbf 1} - 2^{-4|P|} {\mathbf 1}$, by
Theorem 3.12 of \cite{ESY12}. So we have 
$$\norminf{(I-B(\frac{1}{2}(q^*_z+q^*_\vone), \vone))^{-1}} \leq 2^{14|P|+1}$$
Lemma \ref{neednormPPS} now tells us that:
$$\norminf{q^*_\vone - q^*_z} \leq 2^{14|P|+2} \norminf{ \vone- z}$$
This finishes the proof of part $(i)$ of Theorem \ref{errorup}.

To prove part $(ii)$ of Theorem \ref{errorup}, 
first remember that we assume $x=P(x,\vone)$
is strongly connected.
We use part $(ii)$ of Theorem \ref{normbound}.

By assumption, $q^*_\vone=\vone$. We take $z= \frac{1}{2}(\vone + q^*_y)$, giving:

\begin{equation}
\label{eq:bound-in-proof-of-thm-errorup}
\norminf{(I -B(\frac{1}{2}(\vone + q^*_z),\vone))^{-1}} \leq 2^{4|P|} \frac{2}{(\vone-q^*_z )_{\min}}
\end{equation}

\noindent Now 
\begin{eqnarray*}
B(\frac{1}{2}\vone,\vone)(\vone-q^*_z) & \leq & B(\frac{1}{2}(\vone + q^*_z),\vone)(\vone-q^*_z)\\
&  = & P(\vone,\vone) - P(q^*_z, \vone) \quad  \mbox{(by Lemma 3.3 of \cite{ESY12})}  \\
& \leq &  P(\vone,\vone) - P(q^*_z, z) = \vone-q^*_z
\end{eqnarray*}
Now we apply Lemma \ref{lem:bound-on-cone-up-down}, letting $v$ be  $\vone-q^*_z$ in
the statement of that Lemma, and considering $B(\frac{1}{2}\vone,\vone)$ in place of
the $B(\frac{1}{2}\vone)$ in the statement of the Lemma.
This tells us that $\frac{\norminf{\vone-q^*_z}}{(\vone-q^*_z)_{\min}} \leq 2^{|P|}$.

Now, if we substitute this into the equation 
(\ref{eq:bound-in-proof-of-thm-errorup}), we get
$$\norminf{(I -B(\frac{1}{2}(1 + q^*_z), \vone))^{-1}} \leq 2^{5|P| + 1} \frac{1}{\norminf{\vone-q^*_z }}$$
Lemma \ref{neednormPPS} now gives: 
$$\norminf{\vone - q^*_z} \leq 2 \norminf{(I - B(\frac{1}{2}(\vone+q^*_z),\vone))^{-1}} \norminf{1 - z}$$
Inserting our bound for the norm of $(I-B(\frac{1}{2}(\vone+q^*_z), \vone))^{-1}$ gives:
$$\norminf{\vone-q^*_z} \leq 2^{5|P| + 2} \frac{1}{\norminf{\vone-q^*_z }} \norminf{1-z}$$
re-arranging and taking the square root gives:
$$\norminf{\vone-q^*_z} \leq \sqrt{2^{5|P| + 2} \norminf{\vone-z}}$$
As long as the encoding size is $|P| \geq 2$, which we can clearly assume, we have:
 $$\norminf{\vone-q^*_z} \leq 2^{3|P|}\sqrt{\norminf{\vone-z}}$$

For part $(iii)$, the significance of the condition that $\rho(B(\vone,\vone)) < 1$ is that 
it implies $(I - B(\vone,\vone))^{-1}$ exists, and  $(I-B(\vone,\vone))^{-1} \geq (I-B(\frac{1}{2}(\vone+q^*_z),1)$. So, we 
use a bound on $\norminf{(I-B(\vone, \vone))^{-1}}$:

 Lemma \ref{neednormPPS} gives: 
$$\norminf{\vone - q^*_z} \leq 2 \norminf{(I - B(\frac{1}{2}(\vone+q^*_z), \vone))^{-1}} \norminf{\vone - z}$$
Now $ \norminf{(I - B(\frac{1}{2}(\vone+q^*_z), \vone))^{-1}} \leq \norminf{(I-B(\vone, \vone))^{-1}}$. We can apply Lemma \ref{B1normbound} on the PPS $x=P(x,\vone)$, which yields $\norminf{(I-B(\vone, \vone))^{-1}} \leq 2^{3|P|}$. Now we have
$$\norminf{\vone - q^*_z} \leq 2^{3|P|} \norminf{\vone - z}$$
as required.
\qed
\end{proof}

\begin{theorem} \label{deliberatesabotage} Suppose $x=P(x)$ is
a PPS 
in SNF form that
has critical depth at most $\mathfrak{c}$.
Let $\delta \in \real$, such that $0 \leq \delta \leq 2^{-3|P|-1}$.
Suppose that in every  bottom-critical SCC
of $x=P(x)$  
we reduce a single positive coefficient, $p$, by 
setting it to $p' = p (1-\delta)$, 
resulting in the PPS $x=P_\delta(x)$. 
Then
$\norminf{q^* - q^*_\delta } \leq 2^{14|P|+2} \delta^{(1/2^{\mathfrak{c}})}$ where 
$q^*$ and $q^*_\delta$ are the LFP solutions of $x=P(x)$ and $x=P_\delta(x)$, respectively. 
Furthermore,
$\norminf{(I-B_\delta(q^*_\delta))^{-1}} \leq 2^{8|P|+2}\delta^{-3}$.
\end{theorem}  

\begin{proof} If $\mathfrak{c} = 0$, we have no critical SCCs, so we don't change any coefficients, and $q^*= q^*_\delta$, 
and the remaining claim about $\norminf{(I-B_\delta(q^*_\delta))^{-1}}$ follows directly 
from Lemma \ref{nocriticalnorm}.

So, we can assume $\cc > 0$ in the rest of the proof. 
To establish that $q^*$ and $q^*_\delta$ are close, we will use Theorem \ref{errorup}. 
For any SCC, $S$, of a PPS $x=P(x)$, either $q^*_S = \vone$ or $q^*_S < \vone$,
because every variable in $S$ depends (directly or indirectly) on every other, 
so if any of them are $< 1$, then so are
all the others. 

Let $S$ be an SCC with $q^*_S = \vone$ and with $(q^*_\delta)_S < \vone$. 
The SCC $S$ necessarily only depends on SCCs, $T$, with $q^*_T = 1$,
because otherwise we wouldn't have $q^*_S = \vone$.  We want to show that

$$\norminf{\vone - (q^*_\delta)_S} \leq \delta^{(1/2^{\cc_{S \cup D(S)}})} \cdot 2^{6|P_{S \cup D(S)}|} $$

\noindent where $\cc_{S \cup D(S)}$ is the critical depth in $x_{S \cup D(S)} = P_{S \cup D(S)}(x_{S \cup D(S)})$,
and $|P_{S \cup D(S)}|$ denotes the encoding size of the latter PPS. 
To prove this by induction, we can assume
\begin{equation}
\label{eq:induct-claim-core}
\norminf{\vone - (q^*_\delta)_{D(S)}} \leq \delta^{(1/2^{\cc_{D(S)}})} \cdot 2^{6|P_{D(S)}|}
\end{equation}
The base case is when $S$ is a bottom-critical SCC, that does not depend on any other critical SCCs. Then even if $D(S)$ is non-empty, 
$q^*_{D(S)} = (q^*_\delta)_{D(S)}$. 
However, we do change a single coefficient $p$ in $S$, by setting it to $p'= p (1-\delta)$.
Note that because the PPS is in SNF form, $p$ must appear in a equation $x_i = P(x_S,\vone)_i$ where
$x_i$ is of type {\tt L}, and thus the coefficient $p$ appears in a single term $p x_j$.
We wish to consider a new PPS in SNF form, parametrized by the possible values $z \in \{ (1-\delta), 1\}$ that
we multiply $p$ by.  To do this, we can simply add a new variable $x_{n+1}$ (for this particular SCC, $S$),
and we then replace the term $p x_j$ by  $p x_{n+1}$, and we add a new equation $x_{n+1} = z x_j$
to our system of equations.  
We denote this new PPS by $(x_S,x_{n+1})=Q_S((x_S,x_{n+1}), z)$. 
Note that this is indeed a SNF form PPS for either $z \in \{ (1-\delta), 1\}$.
Note also that in terms of encoding size, we have $|Q_S| \leq 2 |P_S|$.

The LFP solution of $(x_S,x_{n+1}) = Q_S((x_S,x_{n+1}),1)$,
in the $S$ coordinates has $q^*_S = \vone$,
and the LFP solution of $(x_S,x_{n+1}) = Q_S((x_S,x_{n+1}),(1-\delta))$
in the $S$ coordinates is $(q^*_\delta)_{S}$. 
Thus, by Theorem \ref{errorup} $(ii)$, we get $\norminf{\vone - (q^*_\delta)_S} \leq 2^{3|Q_S|}\sqrt{\delta} 
\leq  2^{6|P_S|}\sqrt{\delta}$. In this case $\cc_{S \cup D(S)} = 1$ so this is enough to establish
the inductive claim in inequality (\ref{eq:induct-claim-core}).

Next, suppose that $S$ is a critical SCC that depends on a different critical SCC. 
$q^*_S$ is the LFP solution of $x_S=P_S(x_S,q^*_{D(S)})$ and $(q^*_\delta)_S$ is the LFP solution of  $x_S=P_S(x_S,(q^*_\delta)_{D(S)})$. By Theorem \ref{errorup} $(ii)$, $\norminf{\vone - (q^*_\delta)_S} \leq 2^{3|P_S|}\sqrt{\norminf{\vone -(q^*_\delta)_{D(S)})}}$. Substituting using the inductive 
assumption in inequality (\ref{eq:induct-claim-core}) gives: 

\begin{eqnarray*} \norminf{\vone - (q^*_\delta)_S}  & \leq  & 2^{3|P_S|}\sqrt{\norminf{\vone -(q^*_\delta)_{D(S)}}} \\
				& \leq & 2^{3|P_S|} \sqrt{\delta^{(1/2^{\cc_{D(S)}})} 2^{6|P_{D(S)}|}} \\
				& = & 2^{3|P_S| + \frac{6}{2}|P_{D(S)}|} \delta^{(1/2^{\cc_{D(S)} + 1})}   \\
                                & \leq & \delta^{(1/2^{\cc_{S \cup D(S)}})} 2^{3|P_{S \cup D(S)}|}  \end{eqnarray*}

The last inequality holds because
$\cc_{S \cup D(S)} = \cc_{D(S)} + 1$.  This is because $S$ is itself a critical SCC.
Note also that 
$|P_{S \cup D(S)}| = |P_S| + |P_{D(S)}|$ since $x_S=P(x_S,x_{D(S)})_S$ and $x_{D(S)} = P(x_{D(S)})_{D(S)}$ are disjoint subsets of the 
equations in $x=P(x)$.  

Finally suppose that $S$ is not a critical SCC but does have $q^*_S=1$ and depends on some critical SCC. Again $q^*_S$ is the LFP solution of $x_S=P_S(x_S,q^*_{D(S)})$ and $(q^*_\delta)_S$ is the LFP solution of  $x_S=P_S(x_S,(q^*_\delta)_{D(S)})$. 
By Theorem \ref{errorup} $(iii)$: $\norminf{\vone - (q^*_\delta)_S} \leq 2^{3|P_S|} \norminf{\vone -(q^*_\delta)_{D(S)})}$.
Substituting the inductive assumption (\ref{eq:induct-claim-core})
 gives $\norminf{\vone - (q^*_\delta)_S} \leq 2^{3|P_S|+ 6|P_{D(S)}|} \delta^{(1/2^{\cc_{D(S)}})}$ 
which simplifies to $\norminf{\vone - (q^*_\delta)_S} \leq \delta^{(1/2^{\cc_{S \cup D(S)}})} 2^{6|P_{S \cup D(s)}|} $.
This is because $S$ itself is non-critical, so $\cc_{D(S)} = \cc_{S \cup D(S)}$.

Let $A$ (for ``always'') denote the set of variables $x_i$ for which $q^*_i = 1$, 
and let $M$ (for ``maybe'') denote the set of variables $x_i$ for which $0 < q^*_i < 1$. 
$A$ is non-empty as otherwise we would have no critical SCCs. 
Every variable $x_i$ in $A$ is part of some SCC $S$ with $q^*_S = 1$. So our induction 
has already given that
$$\norminf{\vone - (q^*_\delta)_A} \leq {\delta}^{1/2^\cc} 2^{6|P_A|} $$
If $M$ is empty, this bound on $\norminf{q^* - q^*_\delta}$ is enough. Otherwise we have to use 
Theorem \ref{errorup} $(i)$. 
This gives that  $\norminf{q^*_M - (q^*_\delta)_M} \leq 2^{14|P_M|+2} \norminf{\vone - (q^*_\delta)_A}$. 
Substituting gives $\norminf{q^*_M - (q^*_\delta)_M} \leq 2^{14|P|+2} \delta^{1/2^{\cc}}$. 
We have now shown that
$$\norminf{q^* - q^*_\delta } \leq 2^{14|P|+2} {\delta}^{1/2^\cc}$$

The only thing left to complete the proof of Theorem
\ref{deliberatesabotage} is to get a bound on $\norminf{
  (I-B_\delta(q^*_\delta))^{-1} }$.  For this we will use the
techniques of the proof of Theorem \ref{newtonnocritical}.  Call the
set of variables for which $(q^*_\delta)_i = 1$, $A_\delta$ and the
set of variables $x_i$ for which $0 < (q^*_\delta)_i < 1$, $M_\delta$.
Since $q^*_{\delta} \leq q^*$, $M \subseteq M_\delta$ and $A_\delta
\subseteq A$. It is worth noting that variables belonging to critical
SCCs are in $A \cap M_\delta$.  We will first show that if a
variable $x_i$ depends (directly or indirectly) on some variable $x_j$ for which we have
reduced a coefficient in $P_\delta(x)_j$, then $(q^*_\delta)_i \leq 1-
2^{-|P|}\delta$.  For any such $x_i$,
consider a {\em shortest} sequence $x_{l_1}, x_{l_2}, \ldots,
x_{l_m}$, such that (1): $l_1 = j$
and  $P_\delta(x)_j$ has a reduced coefficient in it, 
(2): $l_m = i$, and  (3): for every $0 \leq k < m $,
$P_\delta(x)_{l_{k+1}}$ contains a term with $x_{l_{k}}$.
There is some
term $p_{j,h} x_h$ in $P(x)_j$ which has been changed to $p_{j,h} (1-\delta) x_h$
in $P_\delta(x)_j$.  
Since $x=P(x)$ is a PPS, $P(\vone)_j \leq 1$,
but note that $P_\delta(x)_j$ is not proper,
as indeed we must have that
$P_\delta(\vone)_j \leq P(\vone)_j - p_{j,h}\delta \leq 1 -
p_{j,h}\delta$. 
Also note that $(q^*_\delta)_j = P_\delta(q^*_\delta)_j 
\leq P_\delta(\vone)_j \leq 1 -
p_{j,h}\delta$.  For any $0 \leq k < m $, if $x_{l_{k+1}}$ has type {\tt Q},
then $(q^*_\delta)_{l_{k+1}} \leq (q^*_\delta)_{l_k}$.  If $x_{l_{k+1}}$ has type {\tt L},
then $1 - (q^*_\delta)_{l_{k+1}} \geq
p_{l_{k+1},l_k}(1-(q^*_\delta)_{l_k})$.  By an easy induction $1 - (q^*_\delta)_i
\geq (\prod_{\{k \mid x_{l_k} \text{ has Type L}\}} p_{l_{k+1},l_k} ) (1-
(q^*_\delta)_j)$.  Thus:
$$1 - (q^*_\delta) \geq (\prod_{\{k \mid x_{l_k} \text{ has Type L} \}} p_{l_{k+1},l_k} ) p_{j,h} \delta$$
Since this is the shortest sequence satisfying the stated conditions, for any $0 \leq k < m $, $P_\delta(x)_{l_k}$ has not had any coefficients reduced, and 
furthermore the $x_{l_k}$'s are all distinct variables. So all these coefficients $p_{l_{k+1},l_k}$ and $p_{j,h}$ are distinct coefficients in $x=P(x)$.
The encoding size $|P|$ is at least the number of bits describing these rationals $p_{l_{k+1},l_k}$ and $p_{j,h}$ and thus
$$(q^*_\delta)_i \leq 1- 2^{-|P|} \delta$$

Next we show that the PPS $x=P_\delta(x)$ is non-critical. Suppose, for a contradiction that $x=P_\delta(x)$ is critical. Then it has some critical SCC $S$. But then $S$ 
must have also been an SCC in the PPS $x=P(x)$, because the dependency graphs of these PPSs are the same (we never reduce a positive probability to $0$). 
For $S$ to be a critical SCC in $x=P_\delta(x)$ , we must have that 
$(q^*_\delta)_S = \vone$ and $\rho(B_\delta(\vone)_S) = 1$. 
However, $q^* \geq q^*_\delta$ and $\rho(B(\vone)_S) \geq \rho(B_\delta(\vone)_S) = 1$. 
So $q^*_S=\vone$. 
Lemma 6.5 of \cite{rmc} shows that for any strongly connected PPS, $x=P(x)$,
with Jacobian $B(x)$, and with LFP, $q^*$,  if $x < q^*$, then $\rho(B(x)) < 1$.
Thus, by continuity of eigenvalues, $\rho(B(q^*)) \leq 1$.
Applying this to the strongly connected PPS $x_S =  P(x_S,\vone)_S$,
since $q^*_S=\vone$, we get
$\rho(B(\vone)_S) \leq 1$. Thus $\rho(B(\vone)_S) =1$ i.e. $S$ is a critical SCC of $x=P(x)$. 
Either $S$ is a bottom-critical-SCC or it depends on some bottom-critical-SCC. 
So every variable $x_i$ in $S$ depends on some variable $x_j$ for which we have reduced a coefficient in $P_\delta(x)_j$. So for every $x_i$ in $S$, $q^*_i \leq 1- 2^{-|P|} \delta$. But this contradicts our earlier assertion that $q^*_S = \vone$. 

 $B_\delta(q^*_\delta)$ has the block decomposition $B_\delta(q^*_\delta) = \begin{pmatrix} 
B_\delta(q^*_\delta)_{M_\delta} &  B_\delta(q^*_\delta)_{M_\delta,A_\delta} \\ 0 & B_\delta(q^*_\delta)_{A_\delta} \end{pmatrix}$.

It is possible that $A_\delta$ is empty, in which case the bound we will obtain on $\norminf{(I-B_\delta(q^*_\delta)_{M_\delta})^{-1}}$ will be enough to show the theorem. So we suppose here that $A_\delta$ is non-empty. $M_\delta$ is non-empty since we assumed that we have at least one critical SCC.

We need to show that
both $I-B_\delta(q^*_\delta)_{M_\delta}$ 
and $I-  B_\delta(q^*_\delta)_{A_\delta}$ are
nonsingular, and we need to get upper bounds on $\norminf{(I-B_\delta(q^*_\delta)_{M_\delta})^{-1}}$
and $\norminf{(I-B_\delta(q^*_\delta)_{A_\delta})^{-1}}$.    Once we do so, we can then apply 
Lemma \ref{lem:block-inverse-norm-bound}
to get a bound on $\norminf{(I-B(q^*_\delta))^{-1}}$.

First, let us show that $I-B_\delta(q^*_\delta)_{A_\delta}$ is non-singular,
and also bound  $\norminf{(I-B_\delta(q^*_\delta)_{A_\delta})^{-1}}$.

We note that $P(x)_{A_\delta} = P_\delta(x)_{A_\delta}$. We have shown that any variable $x_i$ for which we have reduced a coefficient in $P_\delta(x)_i$ has $q^*_i \leq 1 - 2^{-|P|}\delta$ and so $x_i$ is not in $A_\delta$. Thus the equations in $x_{A_\delta} = P_\delta(x_{A_\delta})_{A_\delta})$ are a subset of the equations $x=P(x)$ and so the encoding size of this PPS is at most $|P|$. We have also shown that the PPS $x=P_\delta(x)$ is non-critical. So we can apply Lemma \ref{B1normbound} to the PPS 
$x_{A_\delta} = P_\delta(x_{A_\delta})_{A_\delta})$, which gives $\norminf{(I-B_\delta(q^*_\delta)_{A_\delta})^{-1}} \leq 2^{3|P|}$.

Now, let us show that $I-B_\delta(q^*_\delta)_{M_\delta}$ is non-singular,
and also bound  $\norminf{(I-B_\delta(q^*_\delta)_{M_\delta})^{-1}}$.

Consider the PPS, restricted to the variables in $M_\delta$.
Note that no variable in $A_\delta$ can depend on these.
Thus, restricting the PPS $x = P_\delta(x)$ to the variables in $M_\delta$ defines
a PPS  $x_{M_\delta} = P_\delta(x_{M_\delta},\vone)_{M_\delta}$.
Note that the LFP of this is $(q^*_\delta)_{M_\delta} < 1$, by definition of $M_\delta$.
To simplify notation in the current argument,
we shall denote this PPS by  $y = R(y)$, and we shall use $r^* := (q^*_\delta)_{M_\delta}$ 
to denote its LFP.  Furthermore, let us use  $B_R(y)$ to denote its Jacobian.
We note, firstly, that $B_R(r^*) = B_\delta(q^*_\delta)_{M_\delta}$.  
The way to see this is to note that $q^*_\delta = (r^*, \vone)$ and so the entries of both matrices are $\frac{\partial(P_\delta)_i}{\partial x_j}(q^*_\delta)$ for $x_i,x_j \in M_\delta$.

So, rephrased, we want to show $\rho(B_R(r^*)) < 1$, and 
we want to find a bound on $(I- B_R(r^*))^{-1}$.
To do this, we need to follow the 
proof of Theorem \ref{normbound} $(i)$ in the case $y=r^*$. 
(That Theorem was proved in \cite{bmdp}.)

We need to use Lemma \ref{quantpf}, with $A= B_R(r^*)$ and $u= \vone - r^*$. By Lemma 3.5 of 
\cite{ESY12}, $B_R(r^*)(\vone - r^*) \leq \vone - r^*$.
We want to find any $\beta$ so that condition (I) of Lemma \ref{quantpf}
applies to variables $y_i$  such that either $y_i$ has type {\tt Q} or else $R(1)_i < 1$.
Namely for such variables $y_i$, it should be the case that
$(B_R(r^*)(\vone - r^*))_i \leq (1-\beta)(\vone - r^*)_i$.

Let us first note that, for any $y_i$, $r^*_i \leq 1 - 2^{|P|}\delta$. We have shown that if a variable $x_i$ depends on some variable $x_j$ for which we have reduced a coefficient in $P_\delta(x)_j$, then $(q^*_\delta)_i \leq 1- 2^{-|P|}\delta$. If $x_i \in M_\delta$ depends on no such variables, then $x_i \in M$. But then we have $q^*_i \leq 1 - 2^{-4|P|} \leq 1- 2^{-|P|}\delta$ because we assumed that $\delta \leq 2^{-3|P|}$. So for any $x_i \in M_\delta$, $(q^*_\delta)_i \leq  1- 2^{-|P|}\delta$.

In the case  where $y_i=R(y)_i$ has form {\tt Q}, for some $y_j,y_k$, $R(y)_i = y_jy_k$ and so
\begin{eqnarray*} B_r(r^*)(\vone - r^*))_i & = & r^*_j(1-r^*_k) + r^*_k(1-r^*_j) \\
										& = & r^*_j + r^*_k - 2 r^*_jr^*_k \\
										& = & (1-r^*_jr^*_k) - (1 + r^*_jr^*_k - r^*_j - r^*_k) \\
										& = & (1-r^*_i) - (1-r^*_k)(1-r^*_j) \\
										& =  & (1-r^*_i) - \frac{1}{2}((1-r^*_k)(1-r^*_j) + (1-r^*_j) (1-r^*_k))\\
										& \leq & (1-r^*_i) - \frac{1}{2}2^{-|P|}\delta((1-r^*_j) + (1-r^*_k))\\
										& \leq & (1-r^*_i) - \frac{1}{2}2^{-|P|}\delta((1-r^*_j) + (1-r^*_k) - (1-r^*_j) (1-r^*_k))\\
										& = & (1-r^*_i) - \frac{1}{2}2^{-|P|}\delta(1-r^*_i) \\								
& = &(1-\frac{1}{2} 2^{-|P|} \delta)(1 - r^*_i) \end{eqnarray*}

Some variables $x_i$  with $P_\delta(\vone)_i < 1$ have $P(\vone)_i < 1$, in which case $P(\vone)_i \leq 1 - 2^{|P|}$.
 If a variable $x_i$ has $P_\delta(\vone)_i < 1$ but $P(\vone)_i = 1$ then we have reduced some coefficient in $P_\delta(x)_i$ by multiplying it by $1- \delta$ so we have $P_\delta(\vone)_i \leq 2^{-|P|}\delta$. 
 So for any $y_i$ with $R(\vone)_i < 1$, $R(\vone)_i \leq  2^{-|P|}\delta$. So if $R(\vone)_i < 1$,
\begin{eqnarray*} (B_R(r^*)(\vone - r^*))_i  
			& \leq & (B_{R}(\frac{1}{2}(\vone + r^*))(\vone - r^*))_i \\
														& \leq & (R(\vone))_i - (R(r^*))_i \\
														& \leq & (1- 2^{-|P|}\delta) - (r^*)_i \\
														& \leq  & (1 - 2^{-|P}\delta) (1 - q^*_\delta)_i \end{eqnarray*}
														
So condition (I) of Lemma \ref{quantpf},
with $\beta = 2^{-(|P|+1)}\delta$, applies to variables $y_i$ 
which either have type {\tt Q} or have $R_i(\vone) < 1$.

It remains to find an $\alpha$ such that condition (II) of Lemma \ref{quantpf} that applies 
to $y_i$ which either has type {\tt L}  and satisfies $R(\vone)_i = 1$.
(Note that there aren't any variables of type {\tt T} in $M_\delta$, and thus none in $y$.)
 We need the following Lemma from \cite{bmdp}:
\begin{lemma}{(Lemma C.8 of \cite{bmdp})} \label{lem:almost-at-the-end}
 For any PPS, x=P(x), with LFP $0 < q^* < 1$, for
any variable $x_i$ either
\begin{itemize}
\item[(I)]  the equation $x_i = P(x)_i$ is of type {\tt Q}, or else $P(1)_i < 1$.

\item[(II)] $x_i$  depends on a variable $x_j$, such that $x_j=P(x)_i$ is of type {\tt Q},
or else $P(1)_j < 1$.
\end{itemize}\end{lemma}
So given $y_i$ of type {\tt L}  and with $R_i(\vone) = 1$, there is a sequence $y_{l_l}, y_{l_2}, \ldots, y_{l_m}$ with $l_m = i$, with $y_{l_1}$ of type {\tt Q} or $R(\vone)_{l_m} < 1$ and for every $0 \leq k < m $, $R(y)_{l_{k+1}}$ contains a term with $y_{l_{k}}$. 
Without loss of generality, we consider the shortest such sequence. Then for $0 < k \leq m $, $y_{l_k}$ does not have type {\tt Q} so it must have type {\tt L}. Also $R(\vone)_{l_k}=1$. So $R(y)_{l_k}$ contains a term $p_{l_k,l_{k-1}}y_{k-1}$. 
We have that, $B_R(r^*)_{l_k,l_{k-1}} = p_{l_k,l_{k-1}}$. Because $R(\vone)_{l_k}=1$, this term has not been reduced in $P_\delta$, so $p_{l_k,l_{k-1}}$ is a coefficient in $x=P(x)$. 
That this is the shortest sequence implies that each of these is a distinct coefficient in  $x=P(x)$. So $\prod_{k=1}^{m-1} p_{l_{k+1},l_k} \geq 2^{-|P|}$. Now $(B_R(r^*)^{m-1})_{i,l_m} \geq \prod_{k=1}^{m-1} B_R(r^*)_{l_{k+1},l_k} = \prod_{k=1}^{m-1} p_{l_{k+1},l_k} \geq 2^{-|P|}$.

So condition (II) of Lemma \ref{quantpf} applies to $y_i$ of type {\tt L} with $R_i(\vone) = 1$ when $\alpha = 2^{-|P|}$.

We can now use Lemma \ref{quantpf}
with $A= B_R(r^*)$, $u= \vone - r^*$, $\alpha = 2^{-|P|}$ and $\beta = 2^{-|P|}\delta$, giving
$$\norminf{(I-B_R(r^*))^{-1}} \leq \frac{n}{(\vone - r^*)_{\min}^2 2^{-|P|} 2^{-|P|}\delta}$$

 We have argued that $(\vone - r^*)_{\min} \geq 2^{-|P|} \delta$. Using $n \leq 2^{|P|}$ as a (very) conservative bound on $n$, we have: 

\begin{equation}
\label{eq:norm-bound-for-final-M-delta-in-sabotage}
\norminf{(I-B_\delta(q^*_\delta)_{M_\delta})^{-1}} \leq 2^{5|P|}\delta^{-3}
\end{equation}
 
If $A_\delta$ is empty, then $B_\delta(q^*_\delta)= B_\delta(q^*_\delta)_{M_\delta} $ and so we are done.

Otherwise we appeal to Lemma \ref{lem:block-inverse-norm-bound} with the block decomposition $I - B_\delta(q^*_\delta) = \begin{pmatrix} 
I - B_\delta(q^*_\delta)_{M_\delta} &  -B_\delta(q^*_\delta)_{M_\delta,A_\delta} \\ 0 & I - B_\delta(q^*_\delta)_{A_\delta} \end{pmatrix}$.  
Letting $\ZZ = (I-B_\delta(q^*_\delta)_{M_\delta})$, applying Lemma 
\ref{lem:block-inverse-norm-bound}, we get:
 \begin{eqnarray*}
\norminf{(I - B_\delta(q^*_\delta))^{-1}} &\leq & 
\max \{ \norminf{\ZZ^{-1}} + \norminf{\ZZ^{-1}} \norminf{B_\delta(q^*_\delta)_{M_\delta,A_\delta}} \norminf{(I-B_\delta(q^*_\delta)_{A_\delta})^{-1}}, \\
&& \quad \quad \quad \norminf{(I-B_\delta(q^*_\delta)_{A_\delta})^{-1}} \}
\end{eqnarray*}
and
 $\norminf{(I-B_\delta(q^*_\delta)_{A_\delta})^{-1}} \leq 2^{3|P|}$ and $\norminf{B_\delta(q^*_\delta)_{M_\delta,A_\delta}} \leq 2$.
Combining with the bound above in
(\ref{eq:norm-bound-for-final-M-delta-in-sabotage}), we get:
 $$ \norminf{(I - B_\delta(q^*_\delta))^{-1}} \leq \max \{2^{5|P|}\delta^{-3} + 2^{5|P|}\delta^{-3}2^{3|P|}2, 2^{3|P|} \}$$
  Or, more simply, $  \norminf{(I - B_\delta(q^*_\delta))^{-1}} \leq 2^{8|P|+2}\delta^{-3}$.
\qed 
 \end{proof}

\noindent We are finally ready to prove Theorem \ref{thm:main-critical-rdnm}, to which
this entire section was dedicated.\\

\noindent {\bf Theorem 
\ref{thm:main-critical-rdnm}.}
{\em 
For any $\epsilon > 0$, and 
for any SCFG, $G$, in SNF form, with $q^G > 0$, 
with critical depth $\cc(G)$,
consider the new SCFG, $G'$, obtained from $G$ by the following process:  
for each bottom-critical SCC, $\SSS$, of $x = P_G(x)$,  find any rule $r = A \xrightarrow{p} B$ of $G$,
such that $A$ and $B$ are both in $\SSS$  (since $G$ is in SNF, such a rule must exist in every critical SCC).  
Reduce the probability $p$, by setting it to\\ $p' = p (1 - 2^{-(14|G|+3)2^{\cc(G)}} \epsilon^{2^{\cc(G)}})$.
Do this for all bottom-critical SCCs.  This defines $G'$, which is non-critical.

Using $G'$ instead of $G$, if we apply R-NM, with parameter $h+2$ to approximate the LFP solution $q^{G' \otimes D}$ of the MPS 
$y=P_{G' \otimes D}(y)$, then $\norminf{q^{G \otimes D} - x^{[h+1]}} \leq \epsilon$ where 
$h := \lceil \log d  + (3 \cdot 2^{\cc(G)}+1)(\log (1/\epsilon) + 14|G|+3) \rceil$.

 Thus we can compute the probability $q^{G,D}_A = \sum_{t \in F} q^{G \otimes D}_{s_0 A t}$ 
within additive error $\delta > 0$ in time polynomial in: $|G|$, $|D|$, $\log (1/\delta)$, and $2^{\cc(G)}$,
in the standard Turing model of computation.\\
}

\begin{proof}[of Theorem \ref{thm:main-critical-rdnm}]

Note that for an SCFG, $G$, and its corresponding PPS, $x=P_G(x)$,
the bit encoding size of $G$ is at least as big as that of the PPS.
In other words, we have $|G| \geq |P_G|$.
So, we can apply Theorem \ref{deliberatesabotage} to the PPS $x=P_{G}(x)$
with $\delta := 2^{-(14|G|+3)2^{\cc(G)}} \epsilon^{2^{\cc(G)}}$, yielding that
$\norminf{q^G - q^{G'} } \leq \frac{\epsilon}{2}$ and
  $\norminf{(I-B_{G'}(q^{G'}))^{-1}} \leq 2^{8|G|+2 + 3(14|G|+3)2^{\cc{G}}}
  \epsilon^{-3 \cdot 2^{\cc(G)}}$. Now Lemma \ref{bal-lem1} and Lemma
  \ref{bal-properties} $(vi)$  allow us to convert this bound on
  $\norminf{(I-B_{G'}(q^{G'}))^{-1}}$ to a bound on $\norminf{(I-B_{G'
      \otimes D}(q^{G' \otimes D}))^{-1}}$. Namely:

$$\norminf{(I-B_{G' \otimes D}(q^{G' \otimes D}))^{-1}} \leq d 2^{8|G|+2 + 3(14|G|+3)2^{\cc(G)}} \epsilon^{-3 \cdot 2^{\cc(G)}}$$
Now Theorem \ref{normtolconv} gives that $\norminf{ q^*_{G' \otimes D} - x^{[h+1]}} \leq \frac{\epsilon}{2}$ since\\ 
$h \geq \log \norminf{(I-B_{G' \otimes D}(q^*_{G' \otimes D}))^{-1}} + \log (1/\frac{\epsilon}{2})$. Thus  
\begin{eqnarray*} \norminf{q^{G \otimes D} - x^{[h+1]}} & \leq &  \norminf{ q^{G \otimes D} -  q^{G' \otimes D}} 
+ \norminf{ q^{G' \otimes D} - x^{[h+1]}} \\
	& \leq & \norminf{ q^{G} -  q^{G'}} + \norminf{ q^{G' \otimes D} - x^{[h+1] }}  \quad  \mbox{(by Lemma \ref{bal-lem1} \& 
Lemma \ref{bal-properties}$(vi)$)}\\
								& \leq & \frac{\epsilon}{2} + \frac{\epsilon}{2} \\
								& = & \epsilon    \end{eqnarray*}
\qed
\end{proof}

\section{Proof of Proposition \ref{thm:em-not-critical}}

Recall that, for a string $\alpha \in (V \cup \Sigma)^*$, with $n= |V|$, $\kappa(\alpha)$ is the $n$-vector where, for $A \in V$, 
$\kappa_A(\alpha)$ is the number of times $A$ appears in $\alpha$. 
Recall that we define $C(r,\pi)$ to be the number of times the rule $r$ is used 
in the derivation $\pi$, and we define 
$C(A,\pi) = \sum_{r \in R_A} C(r,\pi)$. 
For $A \in V$, define $\mathbf{e}^A$ to be the unit $n$-vector with $(\mathbf{e}^A)_A = 1$ and $(\mathbf{e}^A)_B = 0$ for $B \not= A$.
Define $K(\pi) = \sum_A  C(A, \pi) \mathbf{e}^A$.

Recall that when doing parameter estimation (and EM) we use formula (\ref{eq:learn-SCFG})
$$p(A \rightarrow \gamma) := \frac{ \sum_\pi \PP(\pi) C(A \rightarrow \gamma,\pi) }{\sum_\pi \PP(\pi) C(A,\pi)} $$ 
to obtain (or update) the probabilities of rules in $G$.

Recall that $\PP(\pi)$ is a probability distribution on the complete derivations of
the grammar that start at a designated start nonterminal, $S$.
Again, equation (\ref{eq:learn-SCFG}) only makes sense when
the sums $\sum_\pi \PP(\pi) C(A,\pi)$ are finite and nonzero, which we assume; 
we also assume every non-terminal and rule of $\HH$ appears in some complete derivation $\pi$ with $\PP(\pi) > 0$.\\

\noindent {\bf Proposition \ref{thm:em-not-critical}.}
{\em
If we use parameter estimation to obtain SCFG $G$ using equation (\ref{eq:learn-SCFG}), under the stated assumptions, 
then $G$ is consistent, i.e. $q^G = \vone$, and {\em furthermore}  the PPS $x=P_G(x)$ is non-critical, i.e., $\rho(B_G(\vone)) < 1$.}\\

\noindent A first step toward establishing Proposition \ref{thm:em-not-critical} is the following Lemma, from which we derive a (left) {\em cone vector} for 
$B_G(\vone)$,  which ultimately allows us to show $\rho(B_G(\vone)) < 1$.

\begin{lemma} Let $S$ denote the designated start nonterminal. Then 
$$\mathbf{e}^S= (I - B_G(\vone)^T) (\sum_\pi \PP(\pi) K(\pi))$$
\end{lemma}
\begin{proof}
Firstly, we need to relate $B_G(\vone)$ to the probabilities of the rules.
Given a rule $A \rightarrow \gamma$ we define $B_{A \rightarrow \gamma}(x) := B_{G_{A \rightarrow \gamma}}(x)$ where $G_{A \rightarrow \gamma}$ is an SCFG with the same non-terminals and terminals as $G$ but with only one rule, 
$A \xrightarrow{1} \gamma$, which has probability $1$. 
So then $B_{A \rightarrow \gamma}(\vone)$ is zero outside the $A$ row.
We allow that $G$ may or may not be in normal form. We can say that
$$P_G(x)_A = \sum_{r = (A \rightarrow \gamma) \in R_A} p(r) \prod_{B \in V}  x_B^{\kappa_B(\gamma)}$$
In terms of the ``partial'' SCFGs,  $G_r$, associated with each rule $r \in R$,
this says $P_G(x)_A = \sum_{r \in R_A} p(r) P_{G_{r}}(x)_A$. The $A$ row of 
$B_G(x)$ is then $\sum_{r \in R_A} p(r) B_{r}(x)_A$. Since $B_{A \rightarrow \gamma}(x_G)$ 
is zero outside of the $A$ row, $B_G(x) = \sum_A \sum_{r \in R_A} p(r) B_{r}(x)$. That is:
\begin{equation} B_G(x) = \sum_{r \in R} p(r) B_r(x) \label{eq:moment-matrix-from-rule-probabilities} \end{equation}
 So we can obtain $B_G(\vone)$ from each of the $B_r(\vone)$. 
$B_{A \rightarrow \gamma}(\vone)$ is zero except in the $A$ row. For any non-terminal $B$, \\
$B_{A \rightarrow \gamma}(x)_{A,B} = 
\frac{\partial}{\partial x_B} \prod_C x_C^{\kappa(\gamma)_C} = \kappa_B(\gamma) x_B^{\kappa_B(\gamma) -1} \prod_{C \not= B} x_C^{\kappa(\gamma)_C}$. Evaluated at $\vone$, this yields:
 \begin{equation} (B_{A \rightarrow \gamma}(\vone))_{A,B} =\kappa_B(\gamma) \label{eq:rule-matrix-from-count} \end{equation}

Now we look at what happens to the count of non-terminals in the derivation $\pi$. We have $S \stackrel{\pi}{\Rightarrow} w$ for some $w \in \Sigma^*$. That is, $\pi = r_1 r_2 \ldots r_k \in R^*$, and  
$\alpha_0 \stackrel{r_1}{\Rightarrow} \alpha_1 \stackrel{r_2}{\Rightarrow} 
\alpha_2 \stackrel{r_2}{\Rightarrow} \ldots \stackrel{r_k}{\Rightarrow}
\alpha_m$, for $\alpha_0 = S$, $\alpha_m = w$ and some $\alpha_1, \alpha_2, \ldots, \alpha_{m-1} \in 
(V \cup \Sigma)^*$. 

Consider $\alpha_i \stackrel{r_i}{\Rightarrow} \alpha_{i+1}$ for some $0 \leq i \leq m-1$. The rule $r_i$ is $A_i \rightarrow \gamma_i$ for some non-terminal $A_i$ and some string $\gamma_i$.
Replacing $A_i$ by $\gamma_i$ affects the counts of the non-terminals by $\kappa(\alpha_{i+1})- \kappa(\alpha_i) = \kappa(\gamma_i) - \mathbf{e}^{A_i}$.    
Note that for any nonterminal $A$, and rule $A \rightarrow \gamma$,
we have 
$B_{A \rightarrow \gamma}(\vone)^T \mathbf{e}^{A} = \kappa(\gamma)$, by equation (\ref{eq:rule-matrix-from-count}),
so 
\begin{equation}
\label{eq:for-inside-em-arg}
(I-B_{A \rightarrow \gamma}(\vone)^T)\mathbf{e}^{A} = \mathbf{e}^{A} - \kappa(\gamma)
\end{equation}
Since for any string $w \in \Sigma^*$, we have $\kappa(w) = {\mathbf 0}$, we get:

\begin{eqnarray*} \mathbf{e}^S& = & \mathbf{e}^{S}- \kappa(w)  \\
								& = & \sum_{i=0}^{m-1} \kappa(\alpha_{i}) - \kappa(\alpha_{i + 1}) \\
								& = & \sum_A \sum_{(A \rightarrow \gamma) \in R_A} (C(A \rightarrow \gamma, \pi))(\mathbf{e}^{A} - \kappa(\gamma))  \\
						& = & \sum_A \sum_{(A \rightarrow \gamma) \in R_A} (C(A \rightarrow \gamma, \pi)) (I - B_{A \rightarrow \gamma}(\vone)^T) \mathbf{e}^{A} \quad  \mbox{(by (\ref{eq:for-inside-em-arg}))} \end{eqnarray*}
This is true for any complete derivation $\pi$, so we can use the probability distribution $\PP(\pi)$, 
which has $\sum_\pi \PP(\pi) = 1$ to obtain:
\begin{eqnarray*} \mathbf{e}^S& = & \sum_\pi \PP(\pi)  \sum_A \sum_{(A \rightarrow \gamma) \in R_A} (C(A \rightarrow \gamma, \pi)) (I - B_{A \rightarrow \gamma}(\vone)^T) \mathbf{e}^{A} \\
	& = &  \sum_{A \in V}  ( \sum_{(A \rightarrow \gamma) \in R_A}  \sum_\pi \PP(\pi) (C(A \rightarrow \gamma, \pi)) (I - B_{A \rightarrow \gamma}(\vone)^T))\mathbf{e}^{A} \\
	& = & \sum_A (I - B_G(\vone)^T)(\sum_\pi \PP(\pi) C(A, \pi)) \mathbf{e}^{A} \\
	& = & (I - B_G(\vone)^T) (\sum_\pi \PP(\pi) K(\pi)) \end{eqnarray*}
\qed	
\end{proof}

\begin{proof}[Proof of Theorem \ref{thm:em-not-critical}] Define $v = (\sum_\pi \PP(\pi) K(\pi))$. Then we have that $v = B_G(\vone)^T v + \mathbf{e}^{S}$. We want to use Lemma \ref{quantpf} to show that $\rho(B_G(\vone)^T) < 1$. We can do this by applying it to the vector $u = \frac{1}{\norminf{v}}v$. We do not need explicit bounds 
on $\alpha$, $\beta$ and $u_{\min}$, but we need to show that the conditions hold for some positive $\alpha$, $\beta$ and $u_{\min}$.
Firstly, we note that $v > 0$, since every non-terminal in $G$ appears in some derivation $\pi$ with $\PP(\pi) > 0$. So $u > 0$. Since $u = \frac{1}{\norminf{v}}v$, $\norminf{u} = 1$. 
Note that $u = \frac{1}{\norminf{v}} ( B_G(\vone)^T v + \mathbf{e}^{S}) =  B_G(\vone)^T u +  \frac{1}{\norminf{v}} \mathbf{e}^{S}$.
Thus 
$B_G(\vone)^T u = u - \frac{1}{\norminf{v}} \mathbf{e}^S\leq u$.
In the $S$ coordinate (and only in the $S$ coordinate), we have that $(B_G(\vone)^T u)_S = u_S - \frac{1}{\norminf{v}}  < u_S$, so there is some $\beta > 0$ for which $(B_G(\vone)^T u)_S \leq (1-\beta)u_S$.
For this $\beta$, $u_S$ satisfies condition (I) of Lemma \ref{quantpf}. 
We need to find an $\alpha$ for which all non-terminals other than $S$ satisfy condition (II) of Lemma \ref{quantpf}.

Consider a non-terminal $A \not= S$. $A$ appears in some complete derivation $\pi$ with $\PP(\pi) > 0$.
There is some sequence of (not necessarily consecutive) rules 
$r_{i}: D_i \rightarrow \gamma_i$, $i=1,\ldots, k$,
appearing in that order in $\pi$,  
such that $D_1 = S$, $D_i \in \gamma_{i-1}$ 
for all $2 \leq i \leq k$, and $A \in \gamma_k$.
Without loss of generality $k \leq n$, since otherwise there must be $i,j$ with $2 \leq i < j \leq k$ 
such that $D_i=D_j$ and so the 
shorter sequence $r_{1},...,r_{{i-1}},r_{{j}},...r_{k}$ would have satisfied the above conditions.

\noindent For any $1 \leq i \leq k-1$,  $(B_{r_{i}}(\vone))_{D_i, D_{i+1}} = \kappa(\gamma_i)_{D_{i+1}} \geq 1$, 
and similarly $(B_{r_{k}}(\vone))_{D_k, A} \geq 1$. Now any $r_{j}$, with $1 \leq j \leq k$, appears in $\pi$ which has $\PP(\pi) > 0$. So $p(r_{j}) > 0$.
But  $B_G(\vone) \geq p(r_{j}) B_{r_{j}}(\vone)$. 
So for any $1 \leq i \leq k-1$,  $(B_G(\vone))_{D_i, D_{i+1}} \geq p(r_{i}) > 0$ and 
similarly $B_G(\vone)_{D_k, A} > 0$. So $(B_G(\vone)^k)_{S,A} > 0$.
Then $((B_G(\vone)^T)^k)_{A,S}=((B_G(\vone)^k)^T)_{A,S} = (B_G(\vone)^k)_{S,A} > 0$.
We then define $\alpha_A =  ((B_G(\vone)^T)^k)_{A,S}$. If we take $\alpha = \min_{\{ A \in V \mid  A \not= S\}} \alpha_A$, then $\alpha > 0$ and all 
non-terminals $A \not= S$ satisfy condition (II) of Lemma \ref{quantpf}: i.e., for each $A \not= S$, there is a $k$ with $((B_G(\vone)^T)^k)_{A,S} \geq \alpha$.
We can now apply Lemma \ref{quantpf} which yields that $\rho(B_G(\vone)^T) < 1$. So $\rho(B_G(\vone))  = \rho(B_G(\vone)^T) < 1$.
So, $G$ is not critical.  
Consistency of $G$, i.e., the fact that $q^G = \vone$, also follows.  This holds because, firstly, we can easily see that $G$ is a {\em proper} SCFG.
In other words, for any nonterminal $A$, the sum of the rule probabilities is $1$, because
$\sum_{r \in R_A} p(r) =  \sum_{r \in R_A} \frac{ \sum_\pi \PP(\pi) C(r,\pi) }{\sum_\pi \PP(\pi) C(A,\pi)} = 1$.

Thus, $G$ has a PPS, $x=P_G(x)$, such that $P_G(\vone) = \vone$,  and $\rho(B_G(\vone)) < 1$.
Lemma 6.3 of \cite{rmc} tells us that for any vectors $0 \leq x \leq y$,
$B_G(y)(y-x) \geq P_G(y) - P_G(x)$.
Let $y=\vone$, and let $x = q^G$.    Then we have $B_G(\vone)(\vone - q^G) \geq P_G(\vone) - 
P_G(q^G) = \vone - q^G$, since we have argued both $\vone$ and $q^G$ are fixed points of $P_G$.
But $B_G(\vone)$ is a non-negative square matrix, and $(\vone-q^G) \geq 0$.
Theorem 8.3.2 of \cite{HornJohnson85} tells us that for a square matrix $M \geq 0$,
and vector $v \geq 0$, if $v \neq 0$ and $Mv \geq  v$,
then $\rho(M) \geq 1$.  We know that 
$B_G(\vone)(\vone - q^G) \geq \vone - q^G$,
but we have already established that $\rho(B_G(\vone)) < 1$.
Thus it must be the case that $(\vone - q^G) = 0$.  In other words, $G$ is consistent.
\qed\end{proof}

 \section{A bad example for infix probabilities}

We now present a family of SCFGs, $G_n$, of size $O(n)$, 
and with critical-depth $n$,
and we give a fixed 3-state DFA, $D$.
We use these to indicate why it is likely to be difficult to overcome the 
exponential dependence on critical-depth of the given SCFG, $G$, 
in order to obtain
a P-time algorithms
for computing  the probability (within desired precision) 
that an arbitrary $G$ generates a string in $L(D)$.

The DFA $D$, is depicted in Figure 1.
It has only 3 states and the property it checks is whether $aa$ is an ``infix'' of the string.
In other words, $L(D) = \{ w aa w' \mid w \in \Sigma^* \ \mbox{and} \ w' \in \Sigma^*\}$.  
The family of SCFGs $G_n$ is defined by the following rules:

\begin{figure}[h]
\label{fig:infix-aa}
\centering
\begin{tikzpicture}[->,>=stealth',shorten >=1pt,auto,node distance=2.8cm, semithick]
\tikzstyle{every state}=[circle,fill=black!25,minimum size=17pt,inner sep=0pt]

\node[initial,state] (t1) {$t_1$};
\node[state] (t2) [right of=t1] {$t_2$};
\node[state, accepting] (t3) [right of=t2] {$t_3$};
 
\path
  (t1) edge node{a} (t2)
       edge[loop below] node{b,c} (t1)
  (t2) edge node{a} (t3)
       edge[bend left] node{b,c} (t1)
  (t3) edge[loop right] node{a.b.c} (t3);
\end{tikzpicture}
\caption{Automaton for the infix aa}
\end{figure}
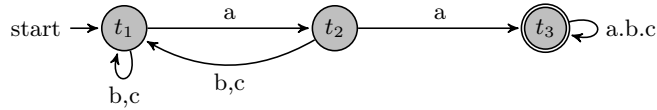

\noindent $A_0 \xrightarrow{0.5} A_0A_0$\\
$A_0 \xrightarrow{0.5} A_1$\\
$A_1 \xrightarrow{0.5} A_1A_1$\\
$A_1 \xrightarrow{0.5} A_2$\\
...\\
$A_n \xrightarrow{1} caB_nac$\\
\\
$B_ n\xrightarrow{1} B_{n-1}B_{n-1}$\\
$B_{n-1}\xrightarrow{1} B_{n-2}B_{n-2}$\\
...\\
$B_0 \xrightarrow{0.5} \epsilon$\\
$B_0 \xrightarrow{0.5} b$\\

\begin{proposition} 
\noindent $q^{G_n} = \vone$. In other words, the probability of termination (generating a finite string) 
starting at any nonterminal in $G_n$ is $1$.

\noindent Furthermore, $q^{G_n \otimes D}_{(t_1A_0t_3)}=\frac{1}{2}$ is the probability that this SCFG $G_n$, 
starting at $A_0$, generates a string which has infix $aa$. 
On the other hand,  $q^{G_n \otimes D}_{(t_1A_it_3)} = 2^{-2^{i}}$
is the same probability, starting at $A_i$.
\end{proposition}

\noindent The proof of this proposition is not at all difficult (using simple induction,
and the formula for solving quadratic equations).  

Let us argue why this causes severe difficulties 
for 
the approximate computation of $q^{G \otimes D}$.
Note that $q^{G_n \otimes D}_{(t_1A_0t_3)}=\frac{1}{2}$
and
$q^{G_n \otimes D}_{(t_1A_nt_3)} = 2^{-2^{n}}$.    However, in the product MPS  $y = P_{G \otimes D}(y)$
the variable $y_{(t_1A_0t_3)}$ depends on the variable $y_{(t_1A_nt_3)}$, and furthermore,
if we, for example, ``under-approximate'' $q^{G_n \otimes D}_{(t_1A_nt_3)} = 2^{-2^{n}}$, and instead 
set $y_{(t_1A_nt_3)} := 0$, or, what effectively achieves the same result, 
if we 
change the product MPS by setting $P_{G \otimes D}(y)_{t_1A_nt_3} \equiv 0$,
then in the resulting modified MPS, with new LFP $\tilde{q}^{G_n \otimes D}$, we would get 
$\tilde{q}^{G_n \otimes D}_{(t_1A_0t_3)} = 0$.

Likewise,  one can show that if we ``over-approximate''  $q^{G_n \otimes D}_{(t_1A_nt_3)}$, 
even very slightly, setting  $P_{G \otimes D}(y)_{t_1A_nt_3} \equiv \frac{1}{2^{poly}}$
in a consistent way,
then we will end up with a new LFP  $\tilde{q}^{G_n \otimes D}$, such that 
$\tilde{q}^{G_n \otimes D}_{(t_1A_0t_3)} \approx  1$  (in other words,
very close to 1).

In both cases, the resulting approximate solution $\tilde{q}^{G_n \otimes D}_{(t_1A_0t_3)}$
is terribly far from the actual solution $\frac{1}{2}$. 
(Note that this is irrespective of the algorithm that is used to compute the
other probabilities.)

Furthermore, we can not in any way use the fact 
that we can detect in P-time and remove
variables $x_A$ from the PPS $x=P_{G_n}(x)$ for which $q^{G_n}_A = 1$, because indeed
$q^G = \vone$, and yet in the product $q^{G \otimes D}$ there are
coordinates with wildly different probabilities that we wish to compute.

\end{document}